\documentclass[11pt,oneside]{amsart}
\usepackage{geometry}                
\usepackage[unicode=true]{hyperref}
\hypersetup{
    colorlinks=true,
    citecolor=cyan,
    linkcolor=magenta,
    urlcolor=magenta}
\usepackage{graphicx}
\usepackage{caption}
\usepackage{datetime}
\usepackage{subcaption} 
\usepackage{amssymb}
\usepackage{epstopdf}
\usepackage{color}
\usepackage{setspace}
\usepackage[authoryear,comma]{natbib}
\usepackage[foot]{amsaddr}
\usepackage{mathtools}
\usepackage{enumitem}
\usepackage{array}
\usepackage{booktabs}
\usepackage{bbm}
\usepackage{comment}
\usepackage{multirow}
\usepackage[foot]{amsaddr}
\usepackage{lipsum}
\usepackage{algorithm}
\usepackage{algpseudocode}
\newcommand{\p}[1]{ {\vec{\mathbf{#1}}} }

\newtheorem{theorem}{Theorem}
\newtheorem{proposition}{Proposition}
\newtheorem{lemma}{Lemma}
\newtheorem{corollary}{Corollary}

\theoremstyle{definition}

\newtheorem{remark}{Remark}
\DeclarePairedDelimiter\ceil{\lceil}{\rceil}
\DeclarePairedDelimiter\floor{\lfloor}{\rfloor}

\usepackage{tikz}
\usetikzlibrary{arrows,calc,shapes,positioning,decorations.pathmorphing,patterns}
\usepackage{pgfplots}

\makeatletter
\renewcommand{\email}[2][]{%
  \ifx\emails\@empty\relax\else{\g@addto@macro\emails{,\space}}\fi%
  \@ifnotempty{#1}{\g@addto@macro\emails{\textrm{(#1)}\space}}%
  \g@addto@macro\emails{#2}%
}

\newcommand\blfootnote[1]{%
  \begingroup
  \renewcommand\thefootnote{}\footnote{#1}%
  \addtocounter{footnote}{-1}%
  \endgroup
}
\makeatother

\definecolor{darkgreen}{RGB}{0,127,0}

\newcolumntype{C}[1]{>{\centering\let\newline\\\arraybackslash\hspace{0pt}}m{#1}}

\begin{document}

\title[]{Best-response dynamics, playing sequences, and convergence to equilibrium in random games}

\date{\today}

\author{Torsten Heinrich$^{1,2,3}$, Yoojin Jang$^{2,4}$, Luca Mungo$^{2,5}$, Marco Pangallo$^{6}$, Alex Scott$^{5}$, Bassel Tarbush$^{7}$, Samuel Wiese$^{2,4}$}
\blfootnote{$^1$Faculty for Economics and Business Administration, Chemnitz University of Technology, Chemnitz, Germany}
\blfootnote{$^2$Institute for New Economic Thinking at the Oxford Martin School, University of Oxford, Oxford, UK}
\blfootnote{$^3$Oxford Martin Programme on Technological and Economic Change (OMPTEC), Oxford Martin School, University of Oxford, Oxford, UK}
\blfootnote{$^4$Department of Computer Science, University of Oxford, Oxford, UK}
\blfootnote{$^5$Mathematical Institute, University of Oxford, Oxford, UK}
\blfootnote{$^6$Institute of Economics and EMbeDS Department, Sant'Anna School of Advanced Studies, Pisa, Italy}
\blfootnote{$^7$Department of Economics, University of Oxford, Oxford, UK}
\blfootnote{\emph{Email addresses}: \texttt{torsten.heinrich@wirtschaft.tu-chemnitz.de}, \texttt{yjluca98@gmail.com}, \texttt{luca.mungo@maths.ox.ac.uk}, \texttt{marco.pangallo@santannapisa.it}, \texttt{scott@maths.ox.ac.uk}, \texttt{bassel.tarbush@economics.ox.ac.uk}, \texttt{samuel.wiese@cs.ox.ac.uk}}
 
\thanks{We thank Doyne Farmer for useful comments at the early stages of this project. We acknowledge funding from Baillie Gifford (Luca Mungo), the James S Mc Donnell Foundation (Marco Pangallo) and the Foundation of German Business (Samuel Wiese). Research supported by EPSRC grant EP/V007327/1 (Alex Scott).}

\ \\
\begin{abstract}
We analyze the performance of the best-response dynamic across all normal-form games using a random games approach. The playing sequence---the order in which players update their actions---is essentially irrelevant in determining whether the dynamic converges to a Nash equilibrium in certain classes of games (e.g. in potential games) but, when evaluated across all possible games, convergence to equilibrium depends on the playing sequence in an extreme way. Our main asymptotic result shows that the best-response dynamic converges to a pure Nash equilibrium in a vanishingly small fraction of all (large) games when players take turns according to a fixed cyclic order. By contrast, when the playing sequence is random, the dynamic converges to a pure Nash equilibrium if one exists in almost all (large) games.

\end{abstract}
\maketitle

\newpage

\onehalfspace

\section{Introduction}
The best-response dynamic is a ubiquitous iterative game-playing process in which, at each time step, players myopically select actions that are a best-response to the actions last chosen by all other players. The literature at large has established the equilibrium convergence properties of the best-response dynamic in games with specific payoff structures; particularly in potential games \citep{monderer1996potential}, but also in weakly acyclic games \citep{fabrikant2013structure}, aggregative games \citep{dindovs2006better}, and quasi-acyclic games \citep{friedman2001learning,takahashi2002pure}. So known results are restricted to special cases. The performance of the best-response dynamic in the class of \emph{all} games remains to be established. In this paper, we consider the question of whether the best-response dynamic converges to a pure Nash equilibrium in a small or large fraction of all possible normal-form games.

To answer our question, we take a ``random games'' approach: we determine whether the best-response dynamic converges to a pure Nash equilibrium in a game drawn at random from among all possible games. The random games approach has a long history in game theory (since \citealp{goldman1957probability,goldberg1968probability}, and \citealp{dresher1970probability}), and has been used to address questions regarding the prevalence of Nash equilibria \citep{powers1990limiting,stanford1995note,stanford1996limit,stanford1997distribution,cohen1998cooperation,stanford1999number,mclennan2005expected,mclennan2005asymptotic,takahashi2008number,kultti2011distribution,daskalakis2011connectivity,quattropani2020efficiency},  the prevalence of rationalizable strategies \citep{pei2019rationalizable}, convergence to equilibrium (\citealp{pangallo2019best}, \citealp*{amiet2021pure}, \citealp*{amiet2021better}, \citealp*{wiese2022frequency}), and the prevalence of dominance solvable games \citep{alon2020dominance}.\footnote{The majority of the literature has focused on normal-form games. \cite{arieli2016random} study random extensive form games.} A guiding principle of the approach is that, since the property of interest (e.g. existence of Nash equilibrium, convergence to Nash equilibrium, dominance solvability) does not hold in all games, one can at least determine how \emph{likely} the property is to hold in the class of all games. To do so, one defines a probability distribution over all games, and computes the probability that a game drawn randomly according to this distribution has the desired property. 

The \emph{playing sequence}---the order in which players update their actions---has an important role in our analysis. We largely focus on two specific playing sequences in this paper. At one extreme, we consider the random playing sequence, where players take turns to play one at a time and the next player to play is chosen uniformly at random from among all players. At the other extreme, we consider a natural deterministic counterpart to the random sequence, which we refer to as the clockwork playing sequence, where players take turns to play one at a time according to a fixed cyclic order. The best-response dynamic under the random playing sequence is widely studied. It is often of interest in population and evolutionary games \citep{sandholm2010population}, and its properties have been analyzed in a variety of games with specific payoff structures.\footnote{It has been analyzed in anonymous games \citep{babichenko2013best}, near-potential games \citep{candogan2013dynamics}, potential games \citep{christodoulou2012convergence,coucheney2014general,swenson2018best,durand2019distributed}, and games on a lattice \citep{blume1993statistical}. ``Sink'' equilibria are studied in \citep{goemans2005sink,mirrokni2009complexity}.} The best-response dynamic under the clockwork playing sequence appears most frequently in the algorithmic game theory literature. Its properties have inter alia been studied in auctions \citep{nisan2011best}, job scheduling \citep{berger2011dynamic}, network formation games \citep{chauhan2017selfish}, and it has been used for equilibrium selection in potential games \citep{boucher2017selecting}. Using the random games approach, \cite{durand2016complexity} show that, in expectation, convergence to equilibrium in potential games is faster under the clockwork playing sequence than under any other playing sequence.

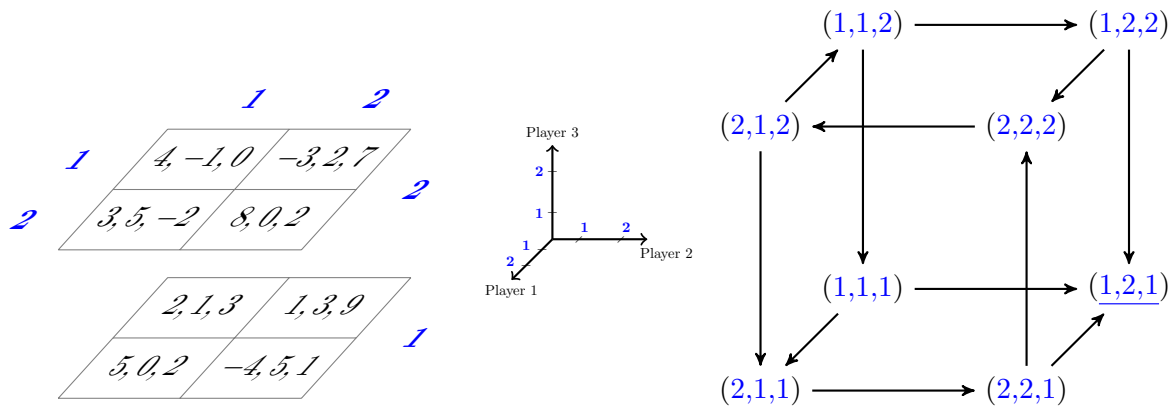
\begin{figure}
\centering
\begin{tikzpicture}
\begin{scope}[scale=0.8,xshift=-50,yshift=70,every node/.append style={yslant=0,xslant=0.9},yslant=0,xslant=0.9]
\draw[xstep=2cm,ystep=1cm,color=gray] (0,0) grid (4,2);
\node at (1,1.5) {$4,-1,0$}; \node at (3,1.5) {$-3,2,7$};
\node at (1,0.5) {$3,5,-2$}; \node at (3,0.5) {$8,0,2$};
\node[blue] at (-1,0.5) {$\mathbf{2}$};
\node[blue] at (-1,1.5) {$\mathbf{1}$};
\node[blue] at (3,2.5) {$\mathbf{2}$};
\node[blue] at (1,2.5) {$\mathbf{1}$};
\node[blue] at (5,1) {$\mathbf{2}$};
\end{scope}

\begin{scope}[scale=0.8,xshift=-50,yshift=0,every node/.append style={yslant=0,xslant=0.9},yslant=0,xslant=0.9]
\draw[xstep=2cm,ystep=1cm,color=gray] (0,0) grid (4,2);
\node at (1,1.5) {$2,1,3$};  \node at (3,1.5) {$1,3,9$};
\node at (1,0.5) {$5,0,2$}; \node at (3,0.5) {$-4,5,1$};
\node[blue] at (5,1) {$\mathbf{1}$};
\end{scope}

\begin{scope}[xshift=200,yshift=-20,
scale=3.5,
>=stealth',
roundnode/.style={rectangle, draw=white, fill=white,inner sep=0,outer sep=0}]   
\foreach \x in {1,2}
\foreach \y in {1,2}
\foreach \z in {1,2}
{\node[] (\z\x\y) at (\x,\y,\z) {({\color{blue}\z,\x,\y})};} 
\path[->] (111) edge [thick ]  (211);
\path[->] (221) edge [thick ]  (121);
\path[->] (111) edge [thick ]  (121);
\path[->] (211) edge [thick ]  (221);

\path[->] (212) edge [thick ]  (112);
\path[->] (122) edge [thick ]  (222);
\path[->] (112) edge [thick ]  (122);
\path[->] (222) edge [thick ]  (212);

\path[->] (212) edge [thick ]  (211);
\path[->] (112) edge [thick ]  (111);
\path[->] (122) edge [thick ]  (121);
\path[->] (221) edge [thick ]  (222);

\node[roundnode] (x) at (121) {({\color{blue} \underline{1,2,1}})};
\end{scope}

\begin{scope}[xshift=145,yshift=60,scale=0.5,every node/.style={transform shape}]
\draw[->,thick] (0,0,0) -- (2.5,0,0) node [pos=1.2,below,rotate=0,yshift=-0.1cm] {Player 2};
\draw[->,thick] (0,0,0) -- (0,2.5,0) node [pos=1,above,rotate=0,yshift=0cm] {Player 3};
\draw[->,thick] (0,0,0) -- (0,0,2.8) node [pos=1,below,rotate=0,yshift=0cm] {Player 1};

\node[rotate=45] at (0.7,0,0) {$-$}; \node[rotate=45] at (1.8,0,0) {$-$};
\node[above right,yshift=0.05cm,xshift=-0.1cm,blue] at (0.7,0,0) {$\mathbf{1}$};
\node[above right,yshift=0.05cm,xshift=-0.1cm,blue] at (1.8,0,0) {$\mathbf{2}$};

\node[rotate=0] at (0,0,0.7) {$-$}; \node[rotate=0] at (0,0,1.8) {$-$};
\node[left,yshift=0.1cm,xshift=-0.17cm,blue] at (0,0,0.7) {$\mathbf{1}$};
\node[left,yshift=0.1cm,xshift=-0.17cm,blue] at (0,0,1.8) {$\mathbf{2}$};

\node[rotate=0] at (0,0.7,0) {$-$}; \node[rotate=0] at (0,1.8,0) {$-$};
\node[left,yshift=0cm,xshift=-0.1cm,blue] at (0,0.7,0) {$\mathbf{1}$};
\node[left,yshift=0cm,xshift=-0.1cm,blue] at (0,1.8,0) {$\mathbf{2}$};
\end{scope}
\end{tikzpicture}
\caption{Illustration of a $3$-player game with $2$ actions per player (left) and its associated best-response digraph (right). The axes shown in the center give us our coordinate system: player 1 selects rows (along the depth), player 2 selects columns (along the width), and player 3 selects levels (along height). In the left-hand panel, the payoffs of players 1, 2, and 3 are listed in that order. The unique pure Nash equilibrium at the profile $(1,2,1)$ is a sink of the digraph and is underlined.}\label{fig:example}
\end{figure}

The playing sequence is essentially irrelevant in determining whether the best-response dynamic converges to equilibrium in potential games---which is the focus of a large part of the literature---but it is a key determinant of the dynamic's convergence properties in non-potential games. To see this, consider the 3-player game shown in the left-hand panel of Figure \ref{fig:example} and its associated \emph{best-response digraph} shown in the right-hand panel. Best-response digraphs are a commonly used reduced-form representation of a game in which the vertices are the action profiles and the directed edges correspond to the players' best-responses (e.g. see \citealp[Chapter 7]{young2020individual}, or \citealp{pangallo2019best}). It is easy to show that potential games have acyclic best-response digraphs, which implies that the playing sequence plays almost no role: as long as each player with a remaining payoff-improving action has a chance to play---which is the case for both the random and the clockwork playing sequences---the dynamic must eventually end at a sink of the digraph, i.e.\ a Nash equilibrium of the game.\footnote{Any such playing sequence may affect the path taken to equilibrium but not whether the path ends at a sink.} In contrast, in the non-potential game shown in Figure \ref{fig:example}, convergence is dependent on the playing sequence: with initial profile $(1,1,1)$, the random sequence best-response dynamic must eventually converge to the Nash equilibrium, whereas the clockwork sequence best-response dynamic with cyclic player order 1-2-3-1-... will remain stuck cycling on the four profiles on the front face of the cube forever. 

Since we are assessing the performance of the best-response dynamic over the class of \emph{all} games (including non-potential games), it is necessary for us to be explicit about the details of the playing sequence. There are, of course, many possible playing sequences,\footnote{The concept of a playing sequence is closely related to ``revision functions'' in \cite{durand2016complexity} and ``schedulers'' in \cite{apt2015classification}. Simultaneous updating by all players at each time step is studied in \cite{quint1997dumb} for 2-player games and in \cite{kash2011multiagent} for anonymous games. \cite{feldman2012convergence} study the case in which the sequence of play depends on current payoffs. \cite{feldman2017efficiency} study the dynamic inefficiency of the best-response dynamic under different playing sequences.} but our focus on random vs.\ clockwork suffices for our main finding: whether the best-response dynamic converges to equilibrium in a small or large fraction of all games depends on the playing sequence in an extreme way. Broadly, we show that under a clockwork playing sequence, the fraction of all $n$-player games  in which the best-response dynamic converges to a pure Nash equilibrium goes to 0 as the number of players and/or actions gets large. By contrast, under a random playing sequence, the fraction of all $n$-player games with a pure Nash equilibrium in which the best-response dynamic converges to a pure Nash equilibrium goes to 1 as the number of players and/or actions gets large (when $n>2$). 

That the best-response dynamic converges less often under a clockwork than under a random playing sequence is perhaps unsurprising since the clockwork sequence will have more difficulty escaping best-response cycles. We therefore expect the probability of convergence to equilibrium for the clockwork sequence to be \textit{less} than it is for the random sequence. However, the resulting extreme jump in the asymptotic equilibrium convergence frequency from 1 to 0 is rather striking. Since most games have digraphs that contain cycles, our contribution can be seen as quantifying the fact that a clockwork playing sequence is \emph{very} likely to become trapped in such cycles, whereas the random playing sequence is \emph{very} likely to escape them.

We now provide a brief technical overview of our methods and results. To generate  games at random, we follow the majority of papers in the 'random games' literature by drawing each player's payoff at each action profile independently according to an arbitrary atomless distribution.\footnote{See \cite{goldberg1968probability,stanford1999number,berg1999entropy,rinott2000number,galla2013complex,sanders2018prevalence,pangallo2019best} for work on random games with payoff correlations.} This induces a uniform distribution over best-response digraphs, and it is in this sense that we can claim convergence in a large or small fraction of all games. The probability of convergence to a pure Nash equilibrium can be reduced to working out the probability that the best-response path initiated at a random vertex hits a sink of the randomly drawn digraph.\footnote{\cite{wiese2022frequency} refer to a game as being ``convergent'' if, from \emph{every} initial vertex, the clockwork best-response dynamic converges to a pure Nash equilibrium. They then show, for each $k \geq 1$, that the probability that a randomly drawn game is convergent and has exactly $k$ pure Nash equilibria is asymptotically zero in large games. Our asymptotic results imply those of \cite{wiese2022frequency}, but not vice versa. Here is why. In Theorem \ref{thm}, we derive upper and lower bounds on the probability that the clockwork best-response dynamic initiated at an arbitrarily chosen vertex converges to a pure Nash equilibrium. But because equilibrium convergence from \emph{every} starting vertex (the focus of \citealp{wiese2022frequency}) implies convergence to equilibrium from an arbitrarily chosen vertex (the focus of our paper) and not conversely, the upper bound that we find in Theorem \ref{thm} is also an upper bound for the probabilities derived in \cite{wiese2022frequency}. Moreover, since we show that the upper bound in Theorem \ref{thm} goes to zero in large games, our asymptotic result  implies the asymptotic results of \cite{wiese2022frequency}, but not vice versa.}

In Section \ref{sec:n>2}, we show that the probability that the clockwork best-response dynamic converges to a pure Nash equilibrium in a game with $n>2$ players and $m_i\geq 2$ actions per player $i$ is, up to a polynomial factor, of order $1/\sqrt{q_{n,\mathbf{m}}}$, where $q_{n,\mathbf{m}} := \frac{\prod_{i=1}^n m_i}{\max_i m_i}$ is the minimal number of strategic environments in the game (i.e. the minimal number of combinations of actions of all but one player). The proof relies on a coupling argument that makes it possible to deal with the path-dependence of the best-response dynamic. The result has two implications. (i) For large $q_{n,\mathbf{m}}$, the probability of convergence is determined by the value of a single parameter, namely, the minimal number of possible strategic environments,  so all games with an identical minimal number of strategic environments have similar asymptotic probabilities of convergence to equilibrium. This is reflected in our simulations even for small values of $q_{n,\mathbf{m}}$. (ii) When the number of players $n$ and/or the number of actions per player is large for at least two players (implying $q_{n,\mathbf{m}}  \rightarrow \infty$), the probability that the clockwork best-response dynamic converges to a pure Nash equilibrium goes to zero. This is in stark contrast with the convergence properties of the random sequence best-response dynamic. 

In Section \ref{sec:n=2}, we provide more detailed theoretical results for games with $n=2$ players. In particular, we provide results on game duration, and we derive an exact expression for the probability that the best-response dynamic converges to a (best-response) cycle of given length at a particular time. As a special case, we obtain the exact probability that the clockwork best-response dynamic converges to a pure Nash equilibrium in 2-player games with $m_i$ actions per player. Unlike in games with $n>2$ players in which the clockwork and random sequences behave very differently from each other, the probability of convergence to equilibrium is the same for the random and clockwork playing sequences in $2$-player games. Furthermore, when $m_1=m_2=m$, we show that this probability is asymptotically $\sqrt{\pi / m}$ when $m$ is large.

Section \ref{sec:simulations} present our simulation results. We investigate the extent to which our asymptotic analytical results also hold for small numbers of players and/or actions. Additionally, we investigate the behavior of playing sequences that interpolate between the extremes of clockwork and random playing sequences.

\section{Best-response dynamics in games}
\label{sec:definitions}

\subsection{Games} A \emph{game} with $n \geq 2$ players and $m_i \geq 2$ actions per player $i$ is a tuple
$$g_{n,\mathbf{m}} := ([n],\{[m_i]\}_{i \in [n]},\{u_i\}_{i \in [n]}), $$
where $\mathbf{m}:=(m_1,...,m_n)$, $[n]:= \{1,...,n\}$ is the set of players, and each player $i \in [n]$ has a set of actions $[m_i]:=\{1,...,m_i\}$ and a payoff function $u_i : \mathcal{M} \rightarrow \mathbb{R}$, where $\mathcal{M}:=\times_{i \in [n]}[m_i]$.

An \emph{action profile} is a vector of actions $\mathbf{a}=(a_1,...,a_n)\in \mathcal{M}$ that lists the action taken by each player. An \emph{environment} for player $i$ is a vector $\mathbf{a}_{-i} \in \mathcal{M}_{-i}:=\times_{j \in [n]\setminus \{i\}}[m_j]$ that lists the action taken by each player but $i$. A \textit{best-response correspondence} $b_i$ for player $i $ is a mapping from the set of environments for player $i$ to the set of all non-empty subsets of $i$'s actions and is defined by
\begin{equation*}
b_i(\mathbf{a}_{-i}) := \arg \max_{a_i \in [m_i]} u_i(a_i, \mathbf{a}_{-i}) .
\end{equation*}
In the rest of this paper, we consider only games in which for each player $i$ and environment $\mathbf{a}_{-i}$, the best-response action is unique. This is the case for games in which there are no ties in payoffs.\footnote{There are no ties in payoffs if for all $i \in [n]$, all $\mathbf{a}_{-i}$, and all $a_i \neq a_i'$, $u_i(a_i ,\mathbf{a}_{-i}) \neq u_i(a_i' ,\mathbf{a}_{-i})$.}

An action profile $\mathbf{a} \in \mathcal{M}$ is a \emph{pure Nash equilibrium} (PNE) if for all $i \in [n]$ and all $a_i \in [m_i]$, $u_i (\mathbf{a}) \geq u_i(a_i,\mathbf{a}_{-i})$. Equivalently, $\mathbf{a}$ is a PNE if each player $i \in [n]$ is playing their (assumed unique) best-response action i.e.\ $a_i = b_i(\mathbf{a}_{-i})$. Denote the set of PNE of the game $g_{n,\mathbf{m}}$ by $\text{PNE}(g_{n,\mathbf{m}})$ and let $\# \text{PNE}(g_{n,\mathbf{m}})$ denote the cardinality of this set.

\subsection{Best-response digraphs}
The best-response structure of a game $g_{n,\mathbf{m}}$ can be represented by a \emph{best-response digraph} $\mathcal{D}(g_{n,\mathbf{m}})$ whose vertex set is the set of action profiles $\mathcal{M}$ and whose edges are constructed as follows: for each $i \in [n]$ and each pair of distinct vertices $\mathbf{a}=(a_i,\mathbf{a}_{-i})$ and $\mathbf{a}' = (a_i',\mathbf{a}_{-i})$, place a directed edge from $\mathbf{a}$ to $\mathbf{a}'$ if and only if $a_i'$ is player $i$'s best-response to environment $\mathbf{a}_{-i}$,  i.e.\ $a_i' = b_i( \mathbf{a}_{-i} )$. There are edges only between action profiles that differ in exactly one coordinate. A profile $\mathbf{a}$ is a PNE of $g_{n,\mathbf{m}}$ if and only if it is a sink of the best-response digraph $\mathcal{D}(g_{n,\mathbf{m}})$. It is easy to show that potential games have acyclic best-response digraphs.\footnote{A game is a (generalized ordinal) potential game if there exists a function $\rho: \mathcal{M} \rightarrow \mathbb{R}$ such that for all $\mathbf{a} \in \mathcal{M}$, $i \in [n]$, and $a_i' \in [m_i]$, $u_i(a_i',\mathbf{a}_{-i}) > u_i(\mathbf{a})$ implies $\rho(a_i',\mathbf{a}_{-i}) > \rho(\mathbf{a})$ \citep{monderer1996potential}. If the best-response digraph has a cycle $(\mathbf{a}_1 \cdots \mathbf{a}_k)$ then the potential function would need to satisfy $\rho(\mathbf{a}_k)>...>\rho(\mathbf{a}_1)>\rho(\mathbf{a}_k)$, a contradiction.}

\subsection{Best-response dynamics}
\label{sec:brdynamic}
We now consider games played over time, with each player in turn myopically best-responding to their current environment.

A \emph{playing sequence} function $s :  \mathbb{N} \rightarrow [n]$ determines whose turn it is to play at each time $t \in \mathbb{N}$, where $\mathbb{N}$ denotes the set of positive integers.\footnote{Our results hold for any permutation of player labels.} We will be interested in two specific playing sequences. The \emph{clockwork} playing sequence is defined by $s_\texttt{c}(t) := 1 + (t-1) \bmod n$, so player 1 plays at time 1, followed by player 2, then 3, and so on until player $n$, and then the sequence returns to player 1, and so on. The \emph{random} playing sequence $s_\texttt{r}$ is determined as follows: for each $t \in \mathbb{N}$, draw $s_\texttt{r}(t)$ uniformly at random from $[n]$. So, at each time, the player playing at that time is drawn uniformly at random from among all players. It is easy to see that, starting from any initial profile, the random sequence best-response dynamic must eventually converge to the PNE of the game shown in Figure \ref{fig:example}, but it is by no means guaranteed to converge to a PNE in all games.\footnote{It is, for example, easy to construct games with a PNE in which there is a cluster of non-PNE profiles that, once visited, cannot be escaped by the random sequence best-response dynamic. \citet*{amiet2021better} refer to such clusters as ``best-response traps''.} In Sections \ref{sec:definitions} and \ref{sec:theory} we restrict our attention to playing sequences $s \in \{s_\texttt{c},s_\texttt{r}\}$.

A \emph{path} $\langle \p{a}, s \rangle$ is an infinite sequence of action profiles $\p{a} =(\mathbf{a}^0, \mathbf{a}^1,...)$ and an associated playing sequence function $s  : \mathbb{N} \rightarrow [n]$ satisfying the constraint that only one player changes her action at a time, i.e.\ $\mathbf{a}_{-s(t)}^t = \mathbf{a}_{-s(t)}^{t-1}$ for each $t \in \mathbb{N}$. So only the action of player $s(t)$ is allowed to differ between profiles $\mathbf{a}^{t-1}$ and $\mathbf{a}^t$ along a path.

The \emph{best-response dynamic} with playing sequence $s : \mathbb{N} \rightarrow [n]$ on a game $g_{n,\mathbf{m}}$ initiated at the action profile $\mathbf{a}^0$ is the following process: set the initial action profile to $\mathbf{a}^0$ and, at each time $t \in \mathbb{N}$, player $s(t)$ myopically plays her best-response $a_i^{t} = b_i(\mathbf{a}_{-i}^{t-1})$ to her current environment $\mathbf{a}_{-s(t)}^{t-1}$. The best-response dynamic effectively generates a path $\langle \p{a}, s \rangle$ by traveling along the edges of the best-response digraph $\mathcal{D}(g_{n,\mathbf{m}})$ in direction $s(t)$ at step $t$ starting from the initial profile $\mathbf{a}^0$.\footnote{More precisely, the infinite sequence of actions $\p{a}$ is determined as follows: if player $s(t)$ is already best responding then $\mathbf{a}^{t-1}$ does not point to any vertex $(a_{s(t)}', \mathbf{a}^{t-1}_{-s(t)}) \neq \mathbf{a}^{t-1}$ and the next profile in the sequence is $\mathbf{a}^{t-1}$ itself, i.e.\ $\mathbf{a}^{t} = \mathbf{a}^{t-1}$; otherwise, if player $s(t)$ is not already playing her best response then travel to the vertex that corresponds to her playing her best-response action, i.e.\ set $\mathbf{a}^{t} = (a_{s(t)}', \mathbf{a}^{t-1}_{-s(t)})$ where $(a_{s(t)}', \mathbf{a}^{t-1}_{-s(t)}) \neq \mathbf{a}^{t-1}$ is the unique vertex that $\mathbf{a}^{t-1}$ points to.}

\subsection{Convergence}\label{subsec:convergence}
For any path $\langle \p{a}, s \rangle$ and set of action profiles $\mathcal{A}\subseteq \mathcal{M}$ the \emph{hitting time} $H_{\langle \p{a}, s \rangle}(\mathcal{A}):=\inf \{t \in \mathbb{N} : \mathbf{a}^t \in \mathcal{A}\}$ is the first time $t\geq 1$ at which some element of the sequence $\p{a}$ is in, or (first) hits, the set $\mathcal{A}$ ($\inf$ is the infimum operator and we use the convention that $\inf \emptyset= \infty$).\footnote{We also say that the path hits $\mathcal{A}$ \emph{by} $t$ if it hits $\mathcal{A}$ at time $\tau$ with $\tau \leq t$, and the path hits $\mathcal{A}$ \emph{before} (\emph{after}) $t$ if it hits $\mathcal{A}$ at time $\tau < t$ ($\tau > t$).}  We say that the $s$-sequence best-response dynamic on game $g_{n,\mathbf{m}}$ initiated at $\mathbf{a}^0$ \emph{converges} to a PNE if its path $\langle \p{a}, s \rangle$ hits $\text{PNE}(g_{n,\mathbf{m}})$ in finite time. Clearly, if a path hits a PNE at some time $t$, it stays there forever after.

\subsection{Best-response dynamics on random games}
We generate random games by drawing all payoffs at random: for each $\mathbf{a} \in \mathcal{M}$ and $i\in [n]$, the payoff $U_i(\mathbf{a})$ is a random number that is drawn from an atomless distribution $\mathbb{P}$. The draws are independent across all $i\in[n]$ and $\mathbf{a} \in \mathcal{M}$. The distribution $\mathbb{P}$ ensures that any ties in payoffs have zero measure, so almost surely each environment has a unique best-response for each player. A random game drawn in this way is denoted by $G_{n,\mathbf{m}} := ( [n] , \{[m_i]\}_{i \in [n]} , \{U_i\}_{i \in [n]} )$.  

The best-response dynamic on random games is described by Algorithm \ref{alg2}. We randomly draw a game and run the best-response dynamic on the drawn game, starting from a randomly drawn initial profile $\mathbf{A}^0$.\footnote{We draw the initial profile $\mathbf{A}^0$ uniformly at random from among all profiles, but this is merely a stylistic choice: since the game itself is drawn at random, the choice of initial condition is actually irrelevant, i.e. our results would not change if we had arbitrarily fixed the initial profile to some specific value.}  Doing so induces a distribution over paths and PNE sets.
\begin{algorithm}
\caption{$s$-sequence best-response dynamic on $G_{n,\mathbf{m}}$\label{alg2}} 
\texttt{
\begin{enumerate}[topsep=0pt,leftmargin=*]
\item For all $i \in [n]$ and $\mathbf{a} \in \mathcal{M}$ draw $U_i(\mathbf{a})$ at random according to $\mathbb{P}$
\item Draw $\mathbf{A}^0$ uniformly at random from $\mathcal{M}$
\item For $t \in \mathbb{N}$:
\begin{enumerate}
\item Set $i=s(t)$
\item Set $\mathbf{A}_{-i}^{t} = \mathbf{A}_{-i}^{t-1}$
\item Set $A_i^t = B_i(\mathbf{A}_{-i}^{t-1})$ where $B_i(\mathbf{A}_{-i}^{t-1}):=\arg \max_{x_i \in [m_i]} U_i(x_i, \mathbf{A}_{-i}^{t-1})$
\end{enumerate}
\end{enumerate}
}
\end{algorithm}

The notion of convergence given in Section \ref{subsec:convergence} applies here. Namely, the $s$-sequence best-response dynamic on game $G_{n,\mathbf{m}}$ (and initial condition $\mathbf{A}^0$) converges to a PNE if its path $\langle \p{A}, s\rangle$ (generated according to Algorithm \ref{alg2}) hits $\text{PNE}(G_{n,\mathbf{m}})$ in finite time.

\section{Theoretical results}
\label{sec:theory}

In this section, we present the theoretical results for best-response dynamics in random games. In Section \ref{sec:n>2} we focus on games with $n>2$ players. In this case, we find that best-response dynamics behave very differently under clockwork vs.\ random playing sequences. Most of our results on the probability of convergence to equilibrium are asymptotic. In Section \ref{sec:n=2} we focus on games with $n=2$ players. In this case, the probability of convergence to equilibrium is the same under both clockwork and random playing sequences. Furthermore, we are able to provide asymptotic as well as \emph{exact} results for game duration and for the probability of convergence to equilibrium.

The quantity
$$q_{n,\mathbf{m}} := \frac{\prod_{i \in [n]} m_i}{\max_{i \in [n]} m_i}$$
is central to our results and it appears frequently in the literature on random games (for example, see \citealp{dresher1970probability}, or \citealp{rinott2000number}). As summarized in the proposition below, the probability that there is a pure Nash equilibrium is asymptotically $1 - \exp\{-1\} \approx 0.63$ as $q_{n,\mathbf{m}}$ gets large.
\begin{proposition}[\citealp{rinott2000number}]\label{prop:NE2}
\begin{equation*}
\lim_{q_{n,\mathbf{m}} \rightarrow \infty} \Pr\left[\# \text{\emph{PNE}}(G_{n,\mathbf{m}}) \geq 1\right] = 1 - \exp\{-1\} .
\end{equation*}
\end{proposition}
\noindent Since $q_{n,\mathbf{m}} \rightarrow \infty$ if and only if $n \rightarrow \infty$ or $m_i \rightarrow \infty$ for at least two players $i$, the probability that there is a PNE in a randomly drawn game approaches $1 - \exp\{-1\}$ when the number of players gets large or when the number of actions per player gets large for at least two players.\footnote{Using results from \cite{arratia1989two}, \cite{rinott2000number} prove the stronger result that the distribution of the number of PNE in random games is asymptotically $\text{Poisson}(1)$ as $q_{n,\mathbf{m}} \rightarrow \infty$. The probability that a PNE exists in a random game was previously studied by \cite{goldberg1968probability} in the 2-player case and by \cite{dresher1970probability} in the $n$-player case as the number of actions gets large for at least two players. \cite{powers1990limiting} and \cite{stanford1995note} noted that the distribution of $\# \text{PNE}(G_{n,\mathbf{m}})$ approaches a Poisson(1) as the number of actions gets large.}

\subsection{Games with $n>2$ players}\label{sec:n>2}
The following result shows that, in large $2$-action games, the random sequence best-response dynamic converges with high probability to a PNE if there is one. Let $\mathbf{2}$ denote a $n$-vector of $2$s.
\begin{proposition}[\citealp*{amiet2021pure}]\label{prop:R1}
\begin{equation*}
\lim_{n \rightarrow \infty} \Pr\left[ s_\texttt{r}\text{\emph{-best-response dynamic on $G_{n,\mathbf{2}}$ converges to a PNE}}  \,|\, \# \text{\emph{PNE}}(G_{n,\mathbf{2}})  \geq 1\right] = 1 .
\end{equation*}
\end{proposition}
\medskip
 Combined with Proposition \ref{prop:NE2}, it follows that over the class of all $2$-action games, the random sequence best-response dynamic converges to a PNE with probability about $(1 - \exp\{-1\})$, i.e.\ in approximately $63\%$ of those games, when the number of players is large.

A generalization of Proposition \ref{prop:R1} to games with more than 2 actions per player is non-trivial. There are currently no existing analytical results for such cases, so this area remains open for future research. However, we conjecture that for $n>2$, the random sequence best-response dynamic converges to a PNE with high probability if there is one as $q_{n,\mathbf{m}} \rightarrow \infty$. Consistent with this conjecture, in the simulations of Section \ref{sec:simulations} we show that, provided $n >2$, the random sequence best-response dynamic does converge to a PNE with probability close to $1-\exp\{-1\}$ when $n$ gets large or when the number of actions gets large for at least two players.

Our main result for the clockwork sequence best-response dynamic in games with $n>2$ players is given below.
\begin{theorem}\label{thm}
\begin{equation*}
\frac{1}{4 \sqrt{n}} \frac{1}{\sqrt{ q_{n,\mathbf{m}} }} \leq 
\Pr\left[
\begin{array}{c}
\text{\emph{$s_\texttt{c}$-best-response dynamic}}\\
\text{\emph{on $G_{n,\mathbf{m}}$ converges to a PNE}}
 \end{array}
 \right] 
\leq  \frac{6 n \sqrt{\log(q_{n,\mathbf{m}})} }{ \sqrt{ q_{n,\mathbf{m}} } } .
\end{equation*}
Consequently, since the upper and lower bound both go to zero as $n$ gets large or when the number of actions gets large for at least two players,
\begin{equation*}
\lim_{q_{n,\mathbf{m}} \rightarrow \infty}\Pr\left[ s_\texttt{c}\text{\emph{-best-response dynamic on $G_{n,\mathbf{m}}$ converges to a PNE}}  \right]  =0 .
\end{equation*}
\end{theorem}
\medskip
So, with high probability, the clockwork sequence best-response dynamic does not converge to a PNE as the number of players gets large or as the number of actions for at least two players gets large. This is in sharp contrast with the asymptotic behavior of the random sequence best-response dynamic. It is intuitive that the clockwork sequence converges to a PNE less often than the random sequence because it will have more difficulty escaping cycles in a best-response digraph. That said, the extreme swing in the asymptotic probability of convergence from 1 to 0 is rather striking.

%

We briefly comment on Theorem \ref{thm} and its implications. (i) In Algorithm \ref{alg2}, drawing payoffs independently at random (from an atomless distribution) induces a uniform distribution over best-response digraphs.\footnote{This follows from the manner in which the payoffs are drawn: there is a zero probability of ties because $\mathbb{P}$ is atomless and for each $i \in [n]$ the probability that action $a_i \in [m_i]$ is a best-response to environment $\mathbf{a}_{-i}$ is given by
$$\Pr\left[U_i(a_i , \mathbf{a}_{-i}) = \max_{x_i \in [m_i]} U_i(x_i,\mathbf{a}_{-i})\right] = \frac{1}{m_i} . $$} It is in this sense that we can say that the best-response dynamic converges in a ``large'' or ``small'' fraction of all games. (ii) Our proof of Theorem \ref{thm} relies on a coupling argument (explained in the appendix) that makes it possible to deal with the path-dependence of the best-response dynamic (which arises from the fact that if a player encounters an environment that they had seen before, they must play the same action that they played when the environment was first encountered). The proof centers on bounding the time it takes for some player to re-encounter a previously seen environment along a best-response path and this time is fundamentally determined by $q_{n,\mathbf{m}}$, which is the minimal number of possible environments. (iii) In fact, Theorem \ref{thm} gives us the following corollary, which shows that the asymptotic probability of convergence to equilibrium is determined primarily by the value of the parameter $q_{n,\mathbf{m}}$.\footnote{Since $\log(q_{n,\mathbf{m}})$ is dominated by a polynomial in $n$ and $\mathbf{m}$, and $q_{n,\mathbf{m}}$ grows faster than $\log(q_{n,\mathbf{m}})$ and than $n$ to any power, the asymptotic behavior of each bound is governed by the behavior of the term $\sqrt{ q_{n,\mathbf{m}} }$ in the denominator.}
\begin{corollary}
\label{corollary:scaling}
The asymptotic probability that the clockwork sequence best-response dynamic converges to a PNE is, up to a polynomial factor, of order $1/\sqrt{ q_{n,\mathbf{m}} }$.
\end{corollary}

\subsection{Games with $n=2$ players}\label{sec:n=2} For $n=2$ players, we provide detailed results on both game duration and on the probability of convergence to equilibrium.  

If the path $\langle \p{a}, s_\texttt{c} \rangle$ generated by the clockwork best-response dynamic on a 2-player game $g_{2,\mathbf{m}}$ has the property that from $t$ onwards, the sequence of $2k$ possibly non-distinct action profiles $\mathbf{a}^t,...,\mathbf{a}^{t+2k-1}$ repeats itself forever and $t$ is the hitting time to $\mathbf{a}^t$, then we say that the clockwork best-response dynamic converged to a cycle of length 2$k$, or a 2$k$-cycle, at time $t$, where $k \in \{1,...,m_*\}$ and $m_*:= \min\{m_1,m_2\}$.

\begin{theorem}\label{thm2}
For any $k \in \{1,...,m_*\}$ and $t \in \{1,...,2(m_*-k+1)\}$,\footnote{For any $k\in \{1,...,m_*\}$ the product is non-negative provided $t + 2k- 2 \leq 2m_*$.}
\begin{align}\label{eq:thm21}
\Pr\left[
\begin{array}{c}
\text{\emph{$s_\texttt{c}$-best-response dynamic on $G_{2,\mathbf{m}}$}}\\
\text{\emph{converges to a $2k$-cycle at time $t$}}
 \end{array}
 \right]=
\frac{1}{m_{s_\texttt{c}(t+2k-1)}}  \prod_{i=1}^{t+2k-2} \left(1 - \frac{1}{m_{s_\texttt{c}(i)}} \floor*{\frac{i}{2}} \right) .  
\end{align}
\end{theorem}
Thus we have an exact expression for the probability that the clockwork sequence best-response dynamic converges to a $2k$-cycle at time $t$.\footnote{See also \cite{pangallo2019best} for an exact formula giving the probability of existence of cycles of any length in 2-player games.} Setting $k=1$ in \eqref{eq:thm21} yields the exact probability that the clockwork sequence best-response dynamic on $G_{2,\mathbf{m}}$ converges to a PNE at time $t$.

As a straightforward corollary of Theorem \ref{thm2}, the probability that the clockwork sequence best-response dynamic converges to a $2k$-cycle is obtained by summing \eqref{eq:thm21} over all $t \in \{1,...,2(m_*-k+1)\}$:
\begin{corollary}\label{cor2}
\begin{equation}\label{eq:n=2}
\Pr\left[
\begin{array}{c}
\text{\emph{$s_\texttt{c}$-best-response}}\\
\text{\emph{dynamic on $G_{2,\mathbf{m}}$}}\\
\text{\emph{converges to a $2k$-cycle}}
 \end{array}
 \right]   = \sum_{t=1}^{2(m_*-k+1)} \; \frac{1}{m_{s_\texttt{c}(t+2k-1)}}  \prod_{i=1}^{t+2k-2} \left(1 - \frac{1}{m_{s_\texttt{c}(i)}} \floor*{\frac{i}{2}} \right)  .
\end{equation}
\end{corollary}
Setting $k=1$ in \eqref{eq:n=2} yields the exact probability that the clockwork sequence best-response dynamic on $G_{2,\mathbf{m}}$ converges to a PNE.

To get a better sense of the behavior of \eqref{eq:n=2}, we now study its asymptotics, which are easiest to see when $m_1=m_2=m$. We maintain this restriction in the rest of this section. Let $\Phi(\cdot)$ denote the standard normal cumulative distribution function:
$$ \Phi(x):= \frac{1}{\sqrt{2\pi}} \int_{-\infty}^x \exp\left\{ - \frac{z^2}{2} \right\} dz. $$
We say that $f(n)$ is asymptotically $g(n)$ if $f(n)/g(n) \rightarrow 1$ as $n \rightarrow \infty$, and $f(n) = o( g(n))$ denotes $f(n)/g(n) \rightarrow 0$ as $n \rightarrow \infty$.
\begin{proposition}\label{prop:F1}
Set $m_1=m_2=m$. If $k = o(m^{2/3})$ then, as $m \rightarrow \infty$, \eqref{eq:n=2} is asymptotically
\begin{equation*}
2 \sqrt{\frac{\pi}{m}} \left( 1 - \Phi\left( \frac{2k -1}{\sqrt{2m}}  \right)\right) .
\end{equation*}
If $k = o(\sqrt{m})$ then, as $m \rightarrow \infty$, \eqref{eq:n=2} is asymptotically $\sqrt{\pi/m}$.
\end{proposition}
\medskip

The asymptotics given in Proposition \ref{prop:F1} help us to better understand the behavior of the clockwork sequence best-response dynamic in large 2-player games. (i) The probability of convergence to a PNE, which corresponds to setting $k =1$, goes to zero when $m \rightarrow \infty$.\footnote{In contrast, for $n=2$, \citet*{amiet2021better} find that ``better''- (rather than best-) response dynamics converge to a PNE (whenever there is one) with high probability as $m \rightarrow \infty$.} (ii) Short cycles all have about the same probability. Indeed, for $k = o(\sqrt{m})$ the probability is asymptotically $\sqrt{\pi/m}$. Finally, (iii) it is very unlikely that the best-response dynamic converges to a very long cycle: if $k/\sqrt{m} \to \infty$ then the probability that the dynamic converges to a cycle of length at least $2k$ tends to 0.\footnote{When $k=o(\sqrt{m})$, the argument of $\Phi(\cdot)$ goes to zero. Since $\Phi(0)=1/2$ we have that the convergence probability goes to $\sqrt{\pi/m}$ which is independent of $k$. If, instead, $k/\sqrt{m} \rightarrow \infty$ then the argument of $\Phi(\cdot)$ grows large and since $\Phi(\infty)=1$, the convergence probability goes zero. Our proof of Proposition \ref{prop:F1} derives the asymptotics for $k = o(m^{2/3})$. The standard normal has small tails outside this range.}

\begin{theorem}\label{thm3}
Set $m_1=m_2=m$ and fix $x>0$. The probability that the $s_\texttt{c}$-best-response dynamic on $G_{2,\mathbf{m}}$ does not hit a cycle (of any length) until at least time step $x\sqrt{2m}$ is asymptotically $\exp\{-x^2/2\}$ as $m \rightarrow \infty$.
\end{theorem}
This result shows that the clockwork sequence best-response dynamic in $2$-player games is likely to converge to a 2$k$-cycle (for some $k \in \{1,...,m\}$) within $\sqrt{2m}$ time steps when $m$ is large.

We now compare the behavior of the clockwork sequence best-response dynamic in 2-player games with the behavior of the random sequence best-response dynamic in 2-player games. (i) The probability of convergence to a PNE is the same for clockwork and for random playing sequences in 2-player games. The reason is that, under the random playing sequence, players' actions do not change whenever the sequence asks the same player to play several times in a row. The profiles that are therefore visited along the path are the same under both playing sequences, which induces the same probability of convergence to equilibrium. However, (ii) the expected game duration will be different since the random playing sequence introduces delays. In fact, the expected game duration for the random playing sequence should be greater than for the clockwork playing sequence by a factor of 2. The reason is that, under the clockwork playing sequence, the players alternate at the tick of each time step, whereas, under the random playing sequence, the time it takes for the playing sequence to turn to the other player is Geometric($\frac{1}{2}$). Thus the random playing sequence can be considered as a slowing down of the clockwork playing sequence in which the expected time to play the next step is 2.

\section{Simulation results}\label{sec:simulations}
In Sections \ref{sec:clock} and \ref{sec:all}, we run simulations of the clockwork and random sequence best-response dynamics. Our main goal is to investigate the extent to which our asymptotic results are also valid for a small number of players and actions. In these simulations, for each choice of $n$ and $\mathbf{m}$, we randomly draw 10 batches of 1000 games. We run the best-response dynamic on each game and find the mean frequency of convergence to equilibrium in each batch, and then report the mean across the batches. The error bars in our figures are intervals of one empirical standard deviation (across the means for each batch).

In Section \ref{sec:periodic} we investigate the equilibrium convergence probability of playing sequences that interpolate between the extremes of clockwork and random playing sequences, and pay particular attention to the speed of convergence.

\subsection{Simulations of clockwork best-response dynamics}\label{sec:clock}
The blue markers in Figure \ref{fig:scaling} show the frequency of convergence to a PNE in our simulations for different values of numbers of players and actions. In both panels, the solid black line is the analytical probability of convergence to a PNE in 2-player $m$-action games, calculated using equation \eqref{eq:n=2} with $m_1=m_2=m$.

In the top panel, we present simulation outcomes for $2$, $3$, and $4$-player games in which all players have the same number of actions. Up to sampling noise, our analytical result for $2$-player games perfectly matches the numerical simulations. We also find that convergence frequency becomes lower for a given number of actions as the number of players increases.

\begin{figure}[htbp]
\centering
\includegraphics[width=\linewidth]{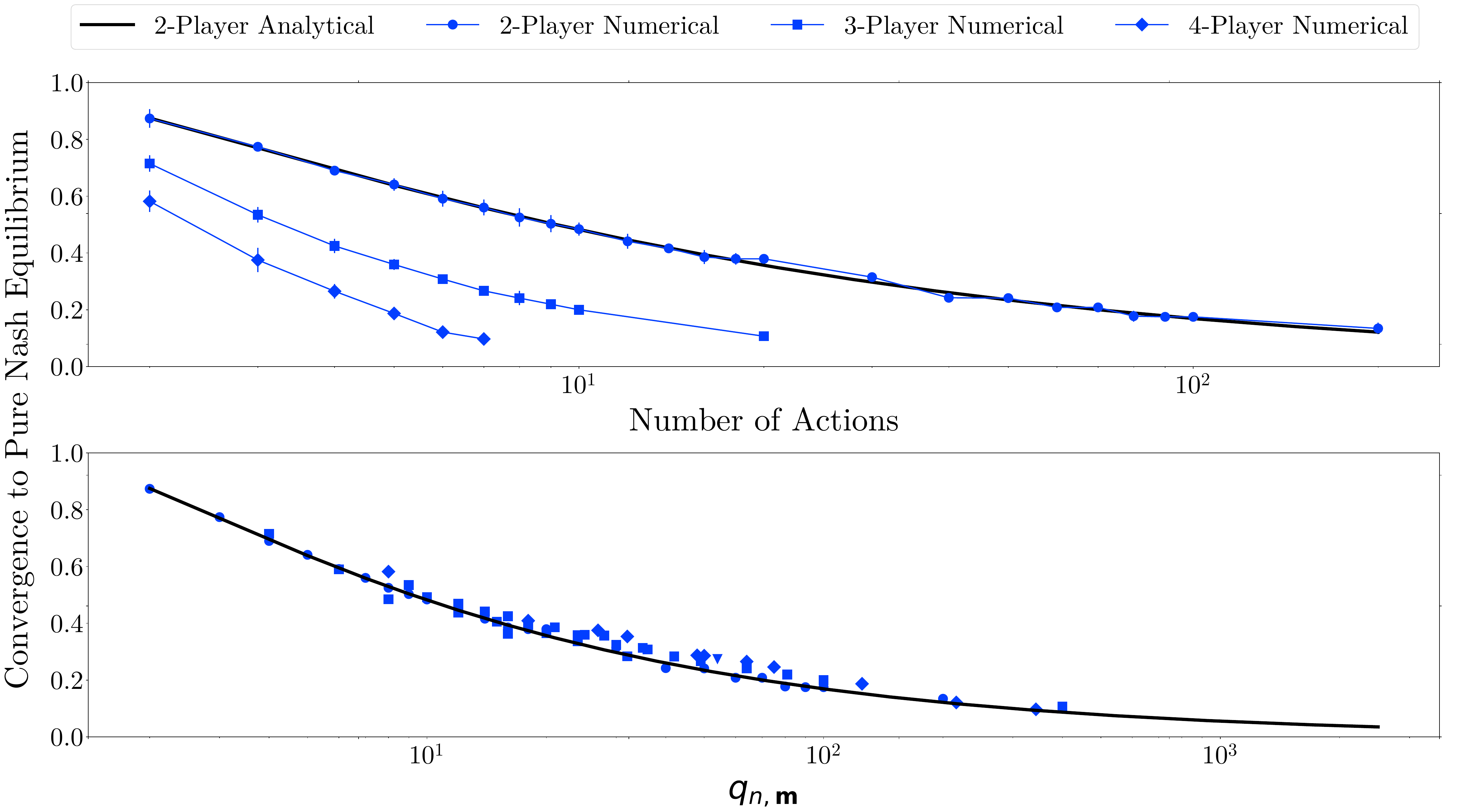}
\caption{Frequency of convergence to a PNE for the clockwork best-response dynamic.}
\label{fig:scaling}
\end{figure}
\begin{figure}[htbp]
\includegraphics[width=\linewidth]{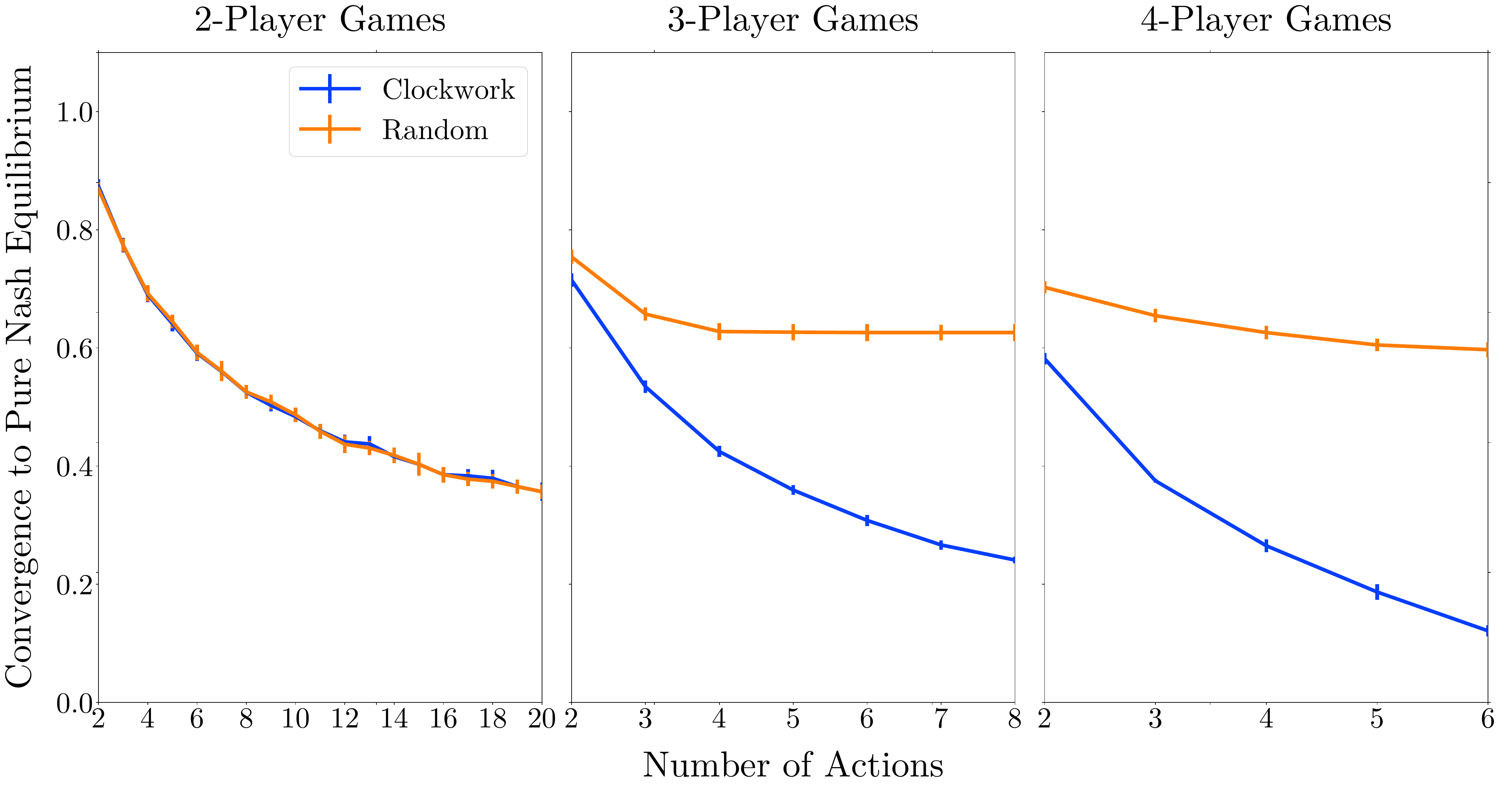}
\caption{Frequency of convergence to a PNE for clockwork vs.\ random best-response dynamics.}
\label{fig:share_only}
\end{figure}

The blue markers in the bottom panel of Figure \ref{fig:scaling} are the simulation means for different values of $n$ and $\mathbf{m}$, all chosen to ensure that the minimal number of environments in those games match the number of environments in a 2-player $m$-action game. All markers line up reasonably well along the solid black line. Corollary \ref{corollary:scaling} implies that the asymptotic convergence probability in games $G_{n,\mathbf{m}}$ and $G_{n',\mathbf{m}'}$ is  approximately the same whenever $q_{n,\mathbf{m}}=q_{n',\mathbf{m}'}$. Our results show that this relation holds even for relatively small games.

\subsection{Simulations of random best-response dynamics}\label{sec:all}

Figure \ref{fig:share_only} shows the frequency of convergence to a PNE under clockwork vs.\ random best-response dynamics in $n$-player games with $m$ actions per player.\footnote{The results also hold if we allow for different numbers of actions per player.}

As argued in Section \ref{sec:n=2}, when there are only $n=2$ players, the random playing sequence has the same convergence probability as the clockwork playing sequence, which can be seen in the left panel of Figure \ref{fig:share_only}. 

Looking across the panels, the frequency of convergence to a PNE is decreasing in both $n$ and $m$ for the clockwork playing sequence, but the random playing sequence is different because its frequency of convergence rapidly settles near $1-1/e$ for $n>2$. Recall, \citet*{amiet2021pure} proved that the random sequence best-response dynamic always converges to a PNE if there is one when $m=2$ and $n\rightarrow \infty$. As argued in Section \ref{sec:n>2}, this gives us an unconditional probability of convergence of  $1-1/e \approx 63\%$. Our simulations show that the result of \citet*{amiet2021pure} also appears to hold for games with more than two actions provided $n>2$. In fact, the random sequence best-response dynamic almost always converges to a PNE in games that have a PNE even for relatively small values of $n$ and $m$.

\subsection{Simulations of periodic best-response dynamics}\label{sec:periodic}
The analytical results of Section \ref{sec:theory} allowed us to compare the behavior of two extreme playing sequences: clockwork and random. We now turn our attention to intermediate cases. A playing sequence is $p$-\emph{periodic} if it consists of a sequence of players of length $p \geq n$ that is repeated forever, with the constraint that each player appears at least once in the repeated sequence. In other words, 
\[
i_1, i_2, ... , i_p, \; i_1, i_2, ..., i_p, \;  i_1, i_2, ..., i_p, \; ...
\]
is a $p$-periodic playing sequence if for each $i \in [n]$ there is some $j \in [p]$ such that $i_j = i$.

We generate $p$-periodic playing sequences at random as follows: construct a sequence of $p-n$ integers drawn at random from $[n]$ and append the numbers $1,...,n$ to this sequence. This results in a sequence of $p$ integers. Now select a random permutation of this sequence and call it $\sigma$. Then $\sigma, \sigma, \sigma,....$ is a $p$-period playing sequence. Clearly, if $p=n$, we recover clockwork playing sequences. And, fixing $n$, we recover random playing sequences for $p \rightarrow \infty$.

\begin{figure}[htbp]
\centering
\includegraphics[width=\linewidth]{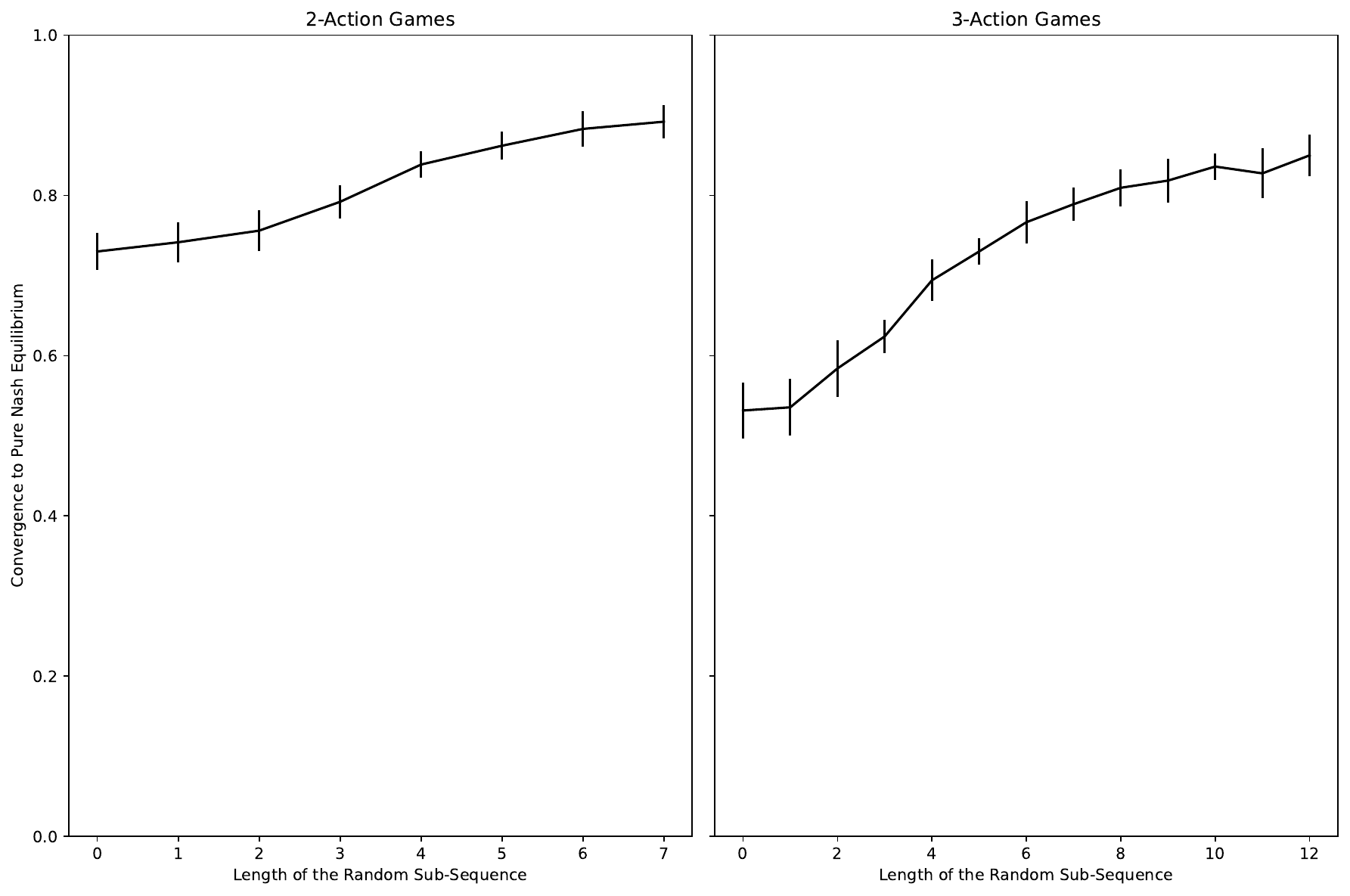}
\caption{Frequency of convergence to a PNE for $p$-periodic best-response dynamics in $n=3$ player games.}
\label{fig:p_share}
\end{figure}
\begin{figure}[htbp]
\includegraphics[width=\linewidth]{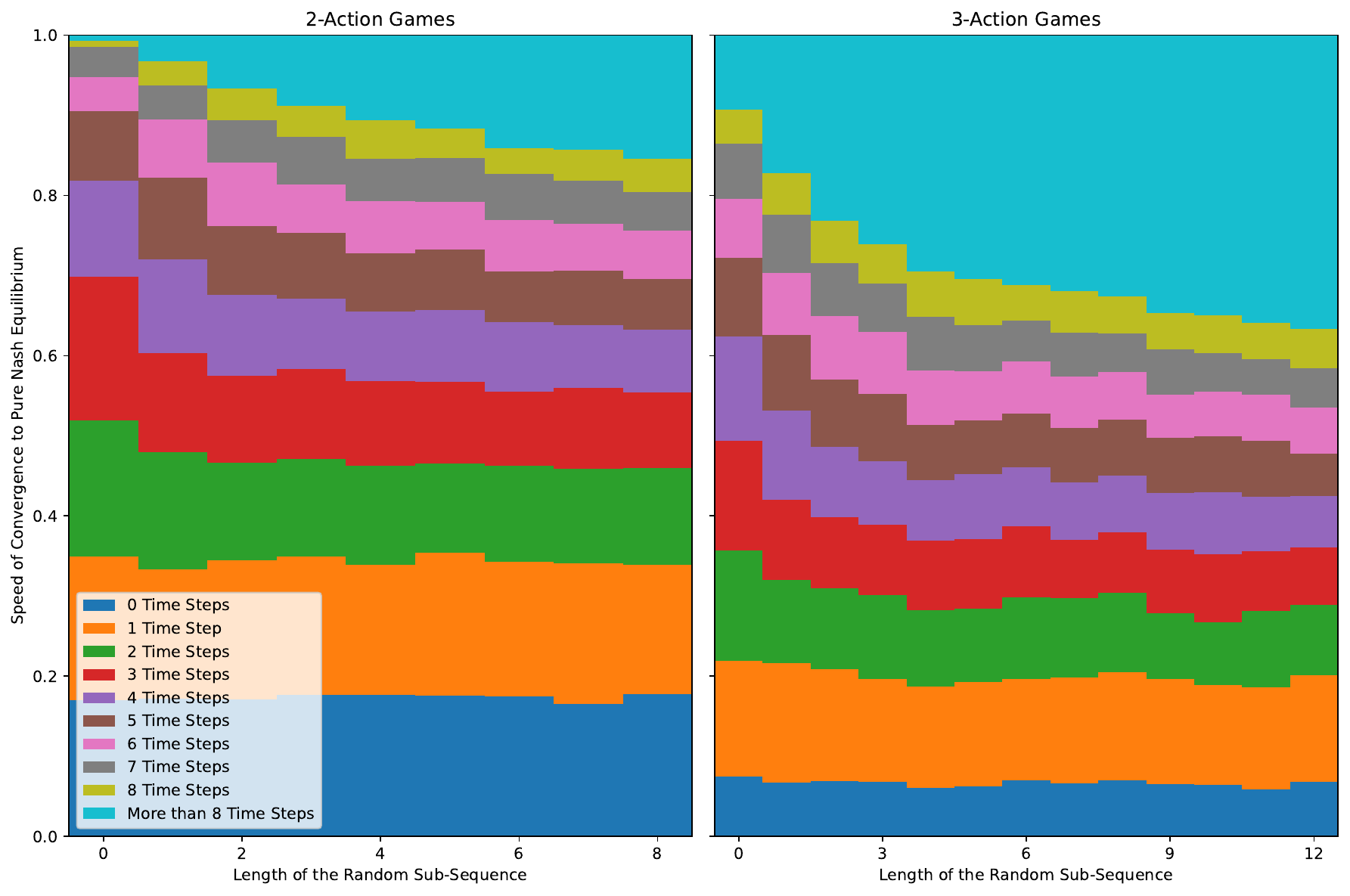}
\caption{Conditional speed of convergence to a PNE for $p$-periodic best-response dynamics in $n=3$ player games.}
\label{fig:p_speed}
\end{figure}

Figure \ref{fig:p_share} plots the frequency of convergence to a pure Nash equilibrium against the length of the random subsequence (namely $p-n$) for $p$-periodic best-response dynamics in $n=3$ player games with $2$ and $3$ actions per player. As might be expected, the probability of convergence to a Nash equilibrium for $p$-periodic playing sequences is increasing in $p$. 

Figure \ref{fig:p_speed} plots, for different lengths of the random subsequence, the distribution of the number of time steps until a pure Nash equilibrium is reached \emph{conditional} on converging to a pure Nash equilibrium for $p$-periodic best-response dynamics in $n=3$ player games with $2$ and $3$ actions per player. Interestingly, conditional on converging to a Nash equilibrium, the average number of time steps to reach equilibrium is increasing in $p$. This relates to the findings of \cite{durand2016complexity} who showed that, in expectation, convergence to equilibrium in potential games is faster under the clockwork playing sequence than under any other playing sequence. Here, our results indicate that, over the space of all games, the speed of convergence (conditional on converging to equilibrium) is slower for playing sequences that have a larger share of random elements (i.e. a larger $p-n$). The apparent trade-off between the success in finding equilibria vs.\ the speed of convergence to equilibria is an interesting area for future research.

\clearpage
\appendix

\section{Proofs}
The appendix concerns only the clockwork best-response dynamic and presents proofs for the results stated in the main body of the paper.

\subsection{Proof of Theorem \ref{thm}}
We start by stating two lemmas that will be used to prove Theorem \ref{thm}. Lemma \ref{prop:after_t} bounds the probability that the clockwork sequence best-response dynamic converges to a pure Nash equilibrium after time $t$. Lemma \ref{prop:by_t} bounds the probability that the clockwork sequence best-response dynamic converges to a pure Nash equilibrium by time $t$.

\begin{lemma}\label{prop:after_t}
Let $\langle \p{A},s_\texttt{c} \rangle$ be generated according to Algorithm \ref{alg2}. For any $t \in \mathbb{N}$,
$$\Pr
\left[ \langle \p{A},s_\texttt{c} \rangle \text{ \emph{hits PNE$(G_{n,\mathbf{m}})$ after} }t \right] \leq \exp\left\{ - \frac{ (\floor*{\frac{t}{n} - 1})^2 }{ 2 q_{n,\mathbf{m}}} \right\} .$$
\end{lemma}

\begin{lemma}\label{prop:by_t}
Let $\langle \p{A},s_\texttt{c} \rangle$ be generated according to Algorithm \ref{alg2}. For any $t \in \mathbb{N}$,
$$\floor*{\frac{t}{n}} \frac{1}{q_{n,\mathbf{m}} }  \left(1- \frac{(\ceil*{\frac{t}{n}})^2}{2} \frac{n}{q_{n,\mathbf{m}} }\right) \leq \Pr \left[ \langle \p{A},s_\texttt{c} \rangle \text{ \emph{hits PNE$(G_{n,\mathbf{m}})$ by} }t  \right] \leq \frac{t}{q_{n,\mathbf{m}}} .$$
\end{lemma}

We now show how Theorem \ref{thm} follows from Lemmas \ref{prop:after_t} and \ref{prop:by_t} and, in Section \ref{sec:lemmas}, we provide proofs for the lemmas themselves.
\begin{proof}[Proof of Theorem \ref{thm}]
Let $\langle \p{A},s_\texttt{c} \rangle$ be generated according to Algorithm \ref{alg2}. The probability that the $s_\texttt{c}$-best-response dynamic on $G_{n,\mathbf{m}}$ converges to a PNE is equal to the probability that $\langle \p{A},s_\texttt{c} \rangle$ hits PNE$(G_{n,\mathbf{m}})$. Let us start with the upper bound. For any $t \in \mathbb{N}$,
\begingroup
\allowdisplaybreaks
\begin{align}
\Pr&\left[ \langle \p{A},s_\texttt{c} \rangle \text{ hits PNE$(G_{n,\mathbf{m}})$}\right] \nonumber \\
&= \Pr\left[ \langle \p{A},s_\texttt{c} \rangle \text{ hits PNE$(G_{n,\mathbf{m}})$ by }t  \right] + \Pr\left[\langle \p{A},s_\texttt{c} \rangle \text{ hits PNE$(G_{n,\mathbf{m}})$ after }t  \right] \nonumber \\
&\leq  \frac{t}{q_{n,\mathbf{m}}} +\exp\left\{ - \frac{ (\floor*{\frac{t}{n} - 1})^2 }{ 2 q_{n,\mathbf{m}} } \right\}. \label{eq:upper} 
\end{align}
\endgroup
Equation \eqref{eq:upper} follows from Lemmas \ref{prop:after_t} and \ref{prop:by_t}. Now, set 
$$t = n \left( \ceil*{ \sqrt{2 q_{n,\mathbf{m}} \log ( q_{n,\mathbf{m}} )} } +1 \right).$$ 
Since $n \geq 2$ and $m_i \geq 2$ for all $i$, we have $\sqrt{2 q_{n,\mathbf{m}} \log ( q_{n,\mathbf{m}} )}>1$, so
$$n \left( \sqrt{2 q_{n,\mathbf{m}} \log ( q_{n,\mathbf{m}} )} + 1 \right) \leq t \leq n \left( \sqrt{2 q_{n,\mathbf{m}} \log ( q_{n,\mathbf{m}} )} + 2 \right) <  3n \sqrt{2 q_{n,\mathbf{m}} \log ( q_{n,\mathbf{m}} )}.$$ 
It follows that 
\begin{equation}\label{eq:ineq1}
\frac{t}{ q_{n,\mathbf{m}} } < 3n \sqrt{\frac{2\log(q_{n,\mathbf{m}})}{q_{n,\mathbf{m}}}} ,
\end{equation}
and that
\begin{equation}\label{eq:ineq2}
\exp\left\{ - \frac{ ( \floor*{ \frac{t}{n} - 1} )^2 }{ 2 q_{n,\mathbf{m}} } \right\}  \leq \frac{1}{ q_{n,\mathbf{m}} } <  n\sqrt{\frac{2\log(q_{n,\mathbf{m}})}{q_{n,\mathbf{m}}}} .
\end{equation}
Adding the upper bounds in \eqref{eq:ineq1} and \eqref{eq:ineq2} yields the desired result.

Let us now turn to the lower bound. For any $t \in \mathbb{N}$,
\begin{align}
\Pr\left[ \langle \p{A},s_\texttt{c} \rangle \text{ hits PNE$(G_{n,\mathbf{m}})$}\right] &\geq \Pr\left[ \langle \p{A},s_\texttt{c} \rangle \text{ hits PNE$(G_{n,\mathbf{m}})$ by }t  \right] \nonumber \\
&\geq \floor*{\frac{t}{n}} \frac{1}{q_{n,\mathbf{m}} }  \left(1- \frac{(\ceil*{\frac{t}{n}})^2}{2} \frac{n}{q_{n,\mathbf{m}} }\right). \label{eq:lower}
\end{align}
Equation \eqref{eq:lower} follows from Lemma \ref{prop:by_t}. Now, set
$$ t = n \floor*{  \frac{\sqrt{ q_{n,\mathbf{m}} }}{ \sqrt{n} }  } .$$
Then,
\begin{equation}\label{eq:ineq3}
1-  \frac{(\ceil{\frac{t}{n}})^2}{2} \frac{n}{q_{n,\mathbf{m}}} \geq \frac{1}{2} .
\end{equation}
And since $n \geq 2$ and $m_i\geq 2$ for all $i$, we have $t \geq \frac{1}{2}\sqrt{n q_{n,\mathbf{m}}} $, so
\begin{equation}\label{eq:ineq4}
\floor*{\frac{t}{n}} \frac{1}{q_{n,\mathbf{m}} } \geq \frac{1}{2 \sqrt{n} } \frac{1}{ \sqrt{ q_{n,\mathbf{m}} } } .
\end{equation}
Multiplying the lower bounds in \eqref{eq:ineq3} and \eqref{eq:ineq4} together yields the desired result.
\end{proof}

\subsection{Lemmas}
\label{sec:lemmas}
We now turn to the proofs of Lemmas \ref{prop:after_t} and \ref{prop:by_t}. These require additional notation which we introduce here.

The notion of convergence given in Section \ref{subsec:convergence} applies to all playing sequences but we can provide a more direct characterization of convergence (and non-convergence) in terms of path properties when the sequence is clockwork. We refer to one complete rotation of the clockwork sequence as a \emph{round} of play; e.g.\ if a round starts at player $i$ then each player plays once in order and the round is complete when it is once again $i$'s turn to play. For any $k \in  \mathbb{N}$ define
$$ T_{\langle \p{a}, s_\texttt{c} \rangle}(k) :=  \inf \left\{t \in \mathbb{N} : \mathbf{a}^t = \mathbf{a}^{t+nk} \text{ and }  \mathbf{a}^t \neq \mathbf{a}^{t+nk'} \text{ for all } k' \in \mathbb{N} \text{ such that } k' < k \right\} , $$
to be the first time at which an action profile is repeated $k$ rounds later (and at no earlier round). If $T_{\langle \p{a}, s_\texttt{c} \rangle}(k)$ is finite, the path $\langle \p{a}, s_\texttt{c} \rangle$ has the property that from time $T_{\langle \p{a}, s_\texttt{c} \rangle}(k)$ onwards, the sequence of $nk$ possibly non-distinct action profiles $\mathbf{a}^t,...,\mathbf{a}^{t+nk-1}$ repeats itself forever, and we say that the path $\langle \p{a}, s_\texttt{c} \rangle$ (first) hits an $nk$-\emph{cycle} at time $T_{\langle \p{a}, s_\texttt{c} \rangle}(k)$. Note that there is exactly one $k$ such that $T_{\langle \p{a}, s_\texttt{c} \rangle}(k)$ is finite.

If the action profile is $\mathbf{a}^t$ at some time $t$ and no one deviates from this profile in a single round (i.e.\ $\mathbf{a}^t = \mathbf{a}^{t+n}$), then $\mathbf{a}^t$ must be a PNE. Therefore, if the path $\langle \p{a}, s_\texttt{c} \rangle$ hits an $nk$-cycle at time $T_{\langle \p{a}, s_\texttt{c} \rangle}(k)$ and $k=1$ ($k>1$) then the clockwork sequence best-response dynamic converges to a PNE (a best-response cycle of length $nk$) at that time. 

Let 
$$ T_{\langle \p{a}, s_\texttt{c} \rangle} :=  \inf \left\{T_{\langle \p{a}, s_\texttt{c} \rangle}(k) : k \in \mathbb{N}  \right\} , $$
denote the first time (necessarily finite) at which the path $\langle \p{a}, s_\texttt{c} \rangle$ hits a PNE or a best-response cycle.

For any path $\langle \p{a}, s_\texttt{c} \rangle$ and for each $t \in \mathbb{N}$ define
$$f_{\langle \p{a}, s_\texttt{c} \rangle}(t) := \min \left\{u \leq t : \mathbf{a}_{-s_\texttt{c}(u)}^{u-1} = \mathbf{a}_{-s_\texttt{c}(t)}^{t-1}  \text{ and }  s_\texttt{c}(u)=s_\texttt{c}(t) \right\} .$$
So $f_{\langle \p{a}, s_\texttt{c} \rangle}(t)$ is the first time along the path $\langle \p{a}, s_\texttt{c} \rangle$ that player $s_\texttt{c}(t)$ encounters the environment $\mathbf{a}_{-s_\texttt{c}(t)}^{t-1}$. Finally, define
$$F_{\langle \p{a}, s_\texttt{c} \rangle}:= \inf \left\{t \in \mathbb{N} : f_{\langle \p{a}, s_\texttt{c} \rangle}(t) < t \right\}.$$
So $F_{\langle \p{a}, s_\texttt{c} \rangle}$ is the first time (necessarily finite) at which some player encounters an environment that they encountered previously along the path.

Remark \ref{rem:3} notes that any path generated by the clockwork best-response dynamic must hit a PNE or a best-response cycle before any player encounters an environment for the second time. 
\begin{remark}\label{rem:3}
$T_{\langle \p{a}, s_\texttt{c} \rangle} < F_{\langle \p{a}, s_\texttt{c} \rangle}$.
\end{remark}
\noindent Roughly speaking, $T_{\langle \p{a}, s_\texttt{c} \rangle}$ denotes the time at which the path $\p{a}$ hits an $nk$-cycle (for some $k \geq 1$) whereas $F_{\langle \p{a}, s_\texttt{c} \rangle}$ denotes the time at which the path completes its first circuit.

The quantities $T_{\langle \p{a}, s_\texttt{c} \rangle}(k)$, $T_{\langle \p{a}, s_\texttt{c} \rangle}$, and  $F_{\langle \p{a}, s_\texttt{c} \rangle}$ are illustrated in an example in Figure \ref{fig:example2}.\ \\

\begin{figure}
\centering
\begin{tabular}{cc}
\begin{subfigure}[b]{0.35\textwidth}
\begin{tikzpicture}[
scale=2.75,
>=stealth',
roundnode/.style={rectangle, draw=white, fill=white,inner sep=0,outer sep=0,minimum size=0.5cm}]   
\foreach \x in {1,2}
\foreach \y in {1,2}
\foreach \z in {1,2}
{\node[] (\z\x\y) at (\x,\y,\z) {$\circ$};} 

\node[roundnode] (112) at (112) {$\mathbf{a}^0,\mathbf{a}^1$};
\node[roundnode] (122) at (122) {$\mathbf{a}^2$};
\node[roundnode] (121) at (121) {$\mathbf{a}^3,\mathbf{a}^4$};

\path[->] (111) edge [thick ]  (211);
\path[->] (221) edge [thick ]  (121);
\path[->] (111) edge [thick ]  (121);
\path[->] (211) edge [thick ]  (221);

\path[->] (212) edge [thick ]  (112);
\path[->] (122) edge [thick ]  (222);
\path[->] (112) edge [thick ]  (122);
\path[->] (222) edge [thick ]  (212);

\path[->] (212) edge [thick ]  (211);
\path[->] (112) edge [thick ]  (111);
\path[->] (122) edge [thick ]  (121);
\path[->] (221) edge [thick ]  (222);
\end{tikzpicture}
\caption{}
\label{fig:gameB}
\end{subfigure}
&
\begin{subfigure}[b]{0.35\textwidth}
\hspace*{-0.6cm}
\begin{tikzpicture}[
scale=2.75,
>=stealth',
roundnode/.style={rectangle, draw=white, fill=white,inner sep=0,outer sep=0,minimum size=0.5cm}]  
\foreach \x in {1,2}
\foreach \y in {1,2}
\foreach \z in {1,2}
{\node[] (\z\x\y) at (\x,\y,\z) {$\circ$};} 

\node[roundnode] (111) at (111) {$\mathbf{a}^0$};
\node[roundnode] (211) at (211) {$\mathbf{a}^1,\mathbf{a}^6,\mathbf{a}^7$};
\node[roundnode] (221) at (221) {$\mathbf{a}^2$};
\node[roundnode] (222) at (222) {$\mathbf{a}^3,\mathbf{a}^4$};
\node[roundnode] (212) at (212) {$\mathbf{a}^5$};

\path[->] (111) edge [thick ]  (211);
\path[->] (221) edge [thick ]  (121);
\path[->] (111) edge [thick ]  (121);
\path[->] (211) edge [thick ]  (221);

\path[->] (212) edge [thick ]  (112);
\path[->] (122) edge [thick ]  (222);
\path[->] (112) edge [thick ]  (122);
\path[->] (222) edge [thick ]  (212);

\path[->] (212) edge [thick ]  (211);
\path[->] (112) edge [thick ]  (111);
\path[->] (122) edge [thick ]  (121);
\path[->] (221) edge [thick ]  (222);
\end{tikzpicture}
\caption{}
\label{fig:gameC}
\end{subfigure}

\end{tabular}
\caption{The digraphs above are identical and correspond to the best-response digraph of the game shown in Figure \ref{fig:example} but we now omit labels to avoid clutter.
In panel (A) the initial profile is set to $\mathbf{a}^0$. The first few elements of the infinite sequence $\p{a}$ are shown. Once at the profile $\mathbf{a}^3$ at $t=3$, which is the unique PNE, the path remains there forever. Here, $T_{\langle \p{a}, s_\texttt{c} \rangle}=T_{\langle \p{a}, s_\texttt{c} \rangle}(1) = 3$ and $F_{\langle \p{a}, s_\texttt{c} \rangle}=6$.
In panel (B) we have a different initial profile $\mathbf{a}^0$. The path moves to the bottom left corner on the front face of the cube at $t= 1$ and then cycles forever among the four profiles on the front face of the cube. In fact, $T_{\langle \p{a}, s_\texttt{c} \rangle}=T_{\langle \p{a}, s_\texttt{c} \rangle}(2) = 1$, so the path hits a $6$-cycle at time 1: once reached, the (not all distinct) action profiles in the sequence $\mathbf{a}^1,...,\mathbf{a}^6$ are repeated forever. Here, $F_{\langle \p{a}, s_\texttt{c} \rangle}=7$.}\label{fig:example2}
\end{figure}
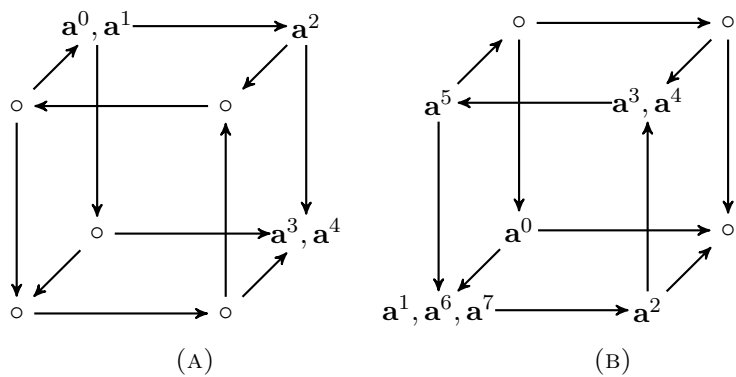

The main challenge posed by paths generated according to Algorithm \ref{alg2} is that they have ``memory'': whenever player $s_\texttt{c}(t)$ encounters an environment that she has encountered before (i.e. $\mathbf{A}_{-s_\texttt{c}(t)}^{t-1} = \mathbf{A}_{-s_\texttt{c}(u)}^{u-1}$ for some $u<t$ with $s_\texttt{c}(t)=s_\texttt{c}(u)$) then, at time $t$, the player must play the same action that she played when she previously encountered the environment (i.e. $A_{s_\texttt{c}(t)}^t = A_{s_\texttt{c}(u)}^u$). This path-dependence complicates the analysis of the clockwork best-response dynamic. We therefore study a simpler (random walk) process that is ``memoryless'' to which we couple a dynamic that induces the same distribution over paths as Algorithm \ref{alg2}. The coupled dynamic follows the random walk process until an environment is encountered by some player for the second time and becomes deterministic thereafter. We elaborate on our argument's reliance on this coupling after the proof of Lemma \ref{prop:after_t}.

The coupled system is described by Algorithms \ref{alg4} and \ref{alg5} and is illustrated in Figure \ref{fig:example4}. $\langle \p{X},s_\texttt{c} \rangle$ and $\langle \p{Y},s_\texttt{c} \rangle$ denote paths generated according to Algorithms \ref{alg4} and \ref{alg5} respectively.
\begin{algorithm}
\caption{Clockwork random walk\label{alg4}} 
\texttt{
\begin{enumerate}[topsep=0pt,leftmargin=*]
\item Draw an initial profile $\mathbf{X}^0$ uniformly at random from $\mathcal{M}$
\item For $t \in \mathbb{N}$:
\begin{enumerate}
\item Set $i=s_\texttt{c}(t)$
\item Set $\mathbf{X}_{-i}^{t} = \mathbf{X}_{-i}^{t-1}$
\item Independently draw $X_i^{t}$ uniformly at random from $[m_i]$
\end{enumerate}
\end{enumerate}
}
\end{algorithm}

\begin{algorithm}
\caption{Coupled dynamic\label{alg5}} 
\texttt{
\begin{enumerate}[topsep=0pt,leftmargin=*]
\item Set $R_i(\mathbf{a}_{-i})=0$ for all $i \in [n]$ and $\mathbf{a}_{-i} \in \mathcal{M}_{-i}$
\item Set the initial action profile to $\mathbf{Y}^0 = \mathbf{X}^0$
\item For $t \in \mathbb{N}$:
\begin{enumerate}
\item Set $i=s_\texttt{c}(t)$
\item Set $\mathbf{Y}_{-i}^{t} = \mathbf{Y}_{-i}^{t-1}$
\item If $R_i(\mathbf{Y}_{-i}^{t-1})=0$: set $Y_i^{t} = X_i^{t}$ and $R_i(\mathbf{Y}_{-i}^{t-1}) = Y_i^t$\\
If $R_i(\mathbf{Y}_{-i}^{t-1})\neq 0$: set $Y_i^{t} = R_i(\mathbf{Y}_{-i}^{t-1})$
\end{enumerate}
\end{enumerate}
}
\end{algorithm}

Algorithm \ref{alg4} is a ``clockwork random walk'' on the set of action profiles $\mathcal{M}$. The walk starts at some randomly drawn initial profile $\mathbf{X}^0$ and, at each time $t$, moves in direction $s_\texttt{c}(t)$ to a profile chosen uniformly at random from among the $m_{s_\texttt{c}(t)}$ profiles in that direction. A path generated according to this process does not have memory.

Algorithm \ref{alg5} describes the coupled dynamic. The process starts at the same initial profile as the clockwork random walk. For each player $i$ and environment $\mathbf{a}_{-i}$, we set the initial ``response'' value $R_i(\mathbf{a}_{-i})$ to zero and update it at step (3c) of the algorithm in the following manner: if the response value to the current environment $\mathbf{Y}_{-i}^{t-1}$ is zero, then the environment was never encountered before and, in that case, player $i$'s response value is set to $X_i^t$, the action drawn by the clockwork random walk at time $t$. If, on the other hand, the response value to the current environment $\mathbf{Y}_{-i}^{t-1}$ is  non-zero (i.e.\ the environment was encountered before), then this value is the action that $i$ takes at time $t$. In other words, $\langle \p{Y},s_\texttt{c} \rangle$ has the same memory property that is characteristic of paths generated according to Algorithm \ref{alg2}.

Algorithm \ref{alg2} essentially draws a best-response digraph ``up-front'', then selects an initial profile and traces a path by traveling along the edges of the digraph starting at the initial profile and moving in direction $s_\texttt{c}(t)$ at step $t$. In contrast, Algorithm \ref{alg5} starts with an empty digraph and then generates its edges in an ``online'' manner. Nevertheless, both algorithms induce the same distribution over paths, as summarized in the following remark.

\begin{remark}\label{rem:1}
Let $\langle \p{A}, s_\texttt{c} \rangle$ and $\langle \p{Y}, s_\texttt{c} \rangle$ be generated according to Algorithms \ref{alg2} and \ref{alg5} respectively. Then $\langle \p{A}, s_\texttt{c} \rangle$ and $\langle \p{Y}, s_\texttt{c} \rangle$ have the same distribution.
\end{remark}

By construction, the sequences $\p{X}$ and $\p{Y}$ must agree at least up to (but not including) the time at which some player encounters an environment for the second time. At such a time, under Algorithm \ref{alg5}, the player must play the action determined by their response function evaluated at that environment but, under Algorithm \ref{alg4}, the next action may be any of the available actions for that player. Remark \ref{rem:2} summarizes the key relationship between the clockwork random walk and the coupled dynamic.
\begin{remark}\label{rem:2}
$F_{\langle \p{X},s_\texttt{c} \rangle} = F_{\langle \p{Y},s_\texttt{c} \rangle}$.
\end{remark}

\begin{figure}
\centering
\begin{tabular}{cc}
\begin{subfigure}[b]{0.35\textwidth}
\hspace*{0cm}
\begin{tikzpicture}[
scale=3,
>=stealth',
roundnode/.style={rectangle, draw=white, fill=white,inner sep=0,outer sep=0,minimum size=0.5cm}]  
\foreach \x in {1,2}
\foreach \y in {1,2}
\foreach \z in {1,2}
{\node[] (\z\x\y) at (\x,\y,\z) {$\circ$};} 

\node[roundnode] (111) at (111) {$\mathbf{X}^0,\mathbf{X}^6,\mathbf{X}^7$};
\node[roundnode] (211) at (211) {$\mathbf{X}^1,\mathbf{X}^2$};
\node[roundnode] (212) at (212) {$\mathbf{X}^3$};
\node[roundnode] (112) at (112) {$\mathbf{X}^4,\mathbf{X}^5$};
\node[roundnode] (121) at (121) {$\mathbf{X}^8$};

\path[-] (111) edge [thick,dotted ]  (211);
\path[-] (221) edge [thick,dotted ] (121);
\path[-] (111) edge [thick,dotted ]  (121);
\path[-] (211) edge [thick,dotted ]  (221);

\path[-] (212) edge [thick,dotted ] (112);
\path[-] (122) edge [thick,dotted ]  (222);
\path[-] (112) edge [thick,dotted ] (122);
\path[-] (222) edge [thick,dotted ]  (212);

\path[-] (211) edge [thick,dotted ] (212);
\path[-] (112) edge [thick,dotted ]  (111);
\path[-] (122) edge [thick,dotted ]  (121);
\path[-] (221) edge [thick,dotted ]   (222);
\end{tikzpicture}
\caption{}
\label{fig:gameC}
\end{subfigure}
&
\begin{subfigure}[b]{0.35\textwidth}
\hspace*{0cm}
\begin{tikzpicture}[
scale=3,
>=stealth',
roundnode/.style={rectangle, draw=white, fill=white,inner sep=0,outer sep=0,minimum size=0.5cm}]  
\foreach \x in {1,2}
\foreach \y in {1,2}
\foreach \z in {1,2}
{\node[] (\z\x\y) at (\x,\y,\z) {$\circ$};} 

\node[roundnode] (111) at (111) {$\mathbf{Y}^0,\mathbf{Y}^6$};
\node[roundnode] (211) at (211) {$\mathbf{Y}^1,\mathbf{Y}^2,\mathbf{Y}^7,\mathbf{Y}^8$};
\node[roundnode] (212) at (212) {$\mathbf{Y}^3$};
\node[roundnode] (112) at (112) {$\mathbf{Y}^4,\mathbf{Y}^5$};

\path[->] (111) edge [thick ]  node[blue,pos=0.5,above left] {1} (211);
\path[-] (221) edge [thick,dotted ] (121);
\path[-] (111) edge [thick,dotted ]  (121);
\path[<-] (211) edge [thick ]  node[blue,pos=0.5,below] {2} (221);

\path[->] (212) edge [thick ] node[blue,pos=0.5,above left] {4} (112);
\path[-] (122) edge [thick,dotted ]  (222);
\path[<-] (112) edge [thick ] node[blue,pos=0.5,above] {5} (122);
\path[-] (222) edge [thick,dotted ]  (212);

\path[->] (211) edge [thick ] node[blue,pos=0.5,left] {3} (212);
\path[->] (112) edge [thick ]  node[blue,pos=0.5,right] {6} (111);
\path[-] (122) edge [thick,dotted ]  (121);
\path[-] (221) edge [thick,dotted ]   (222);
\end{tikzpicture}
\caption{}
\label{fig:gameC}
\end{subfigure}
\end{tabular}
\caption{Illustration of Algorithms \ref{alg4} and \ref{alg5}. Panel (A) shows the first few elements of a possible path $\langle \p{X},s_\texttt{c} \rangle$ generated according to the clockwork random walk starting at the profile $\mathbf{X}^0$. Panel (B), illustrates the first few elements of the corresponding path $\langle \p{Y},s_\texttt{c} \rangle$ generated according to Algorithm \ref{alg5}, starting with an empty digraph and numbering the directed edges according to the time at which they are first placed. 
The paths in panels (A) and (B) are identical up to and including time 6. At time step 7, however, player 1 encounters the same environment that she had encountered at time 1 ($F_{\langle \p{X},s_\texttt{c} \rangle} = F_{\langle \p{Y},s_\texttt{c} \rangle}=7$); namely, players 2 and 3 each choosing action 1. The first time that player 1 encountered this environment, she responded by playing action 2, so she must play action 2 again at time 7. From then on, the path in panel (B) will keep cycling among the action profiles on the left-hand side of the cube forever whereas the path in panel (A) is allowed to wander freely.}\label{fig:example4}
\end{figure}
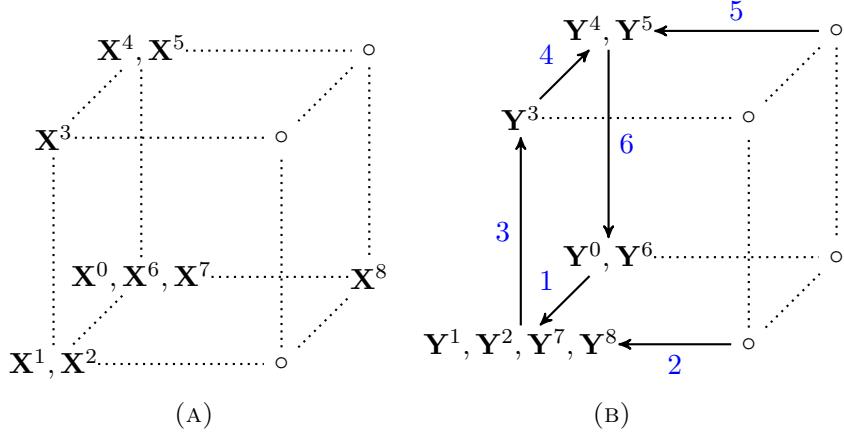

The lemma below, which concerns paths $\langle \p{X},s_\texttt{c} \rangle$ that are generated by the clockwork random walk, is useful for proving Lemmas \ref{prop:after_t} and \ref{prop:by_t}. Under the clockwork sequence, player $i \in [n]$ plays at time $h_i(k) := i + (k-1)n$ for $k \in \mathbb{N}$. For any $i \in [n]$ and any time $t \geq i$, define 
$$k^*_i (t) := 1 + \floor*{\frac{t-i}{n}} .$$
So $k^*_i (t)$ is the largest $k \in \mathbb{N}$ such that $h_i(k) \leq t$. Between times 1 and $t$ (inclusive), player $i \in [n]$ plays at times $h_i(1),h_i(2),...,h_i(k^*_i(t))$ and encounters environments $\mathbf{X}_{-i}^{h_i(1) - 1},\mathbf{X}_{-i}^{h_i(2) - 1},...,\mathbf{X}_{-i}^{h_i(k^*_i(t))-1}$. Lemma \ref{lem:bounds} establishes bounds on the probability that these environments are all distinct.

Define $\mu:=\prod_{i \in [n]}m_i$ to be the cardinality of $\mathcal{M}$.
\begin{lemma}\label{lem:bounds}
For any $i \in[n] $ and $t \in \mathbb{N}$,
\begin{equation*}
1 - \frac{m_i}{\mu} \frac{(\ceil*{\frac{t}{n}})^2}{2} \leq \Pr\left[\mathbf{X}_{-i}^{h_i(k) - 1} \text{ \emph{for} } k \in \{1,...,k^*_i(t)\} \text{ \emph{are all distinct}}\right] \leq \exp\left\{ - \frac{ m_i(\floor*{\frac{t}{n} - 1})^2 }{ 2\mu } \right\} .
\end{equation*}
\end{lemma}
\begin{proof}
For any $i \in [n]$, the environments $\mathbf{X}_{-i}^{h_i(1) - 1}, \mathbf{X}_{-i}^{h_i(2) - 1},...,\mathbf{X}_{-i}^{h_i(k^*_i(t)) - 1}$ are independent because they are disjoint subsets of the draws of the clockwork random walk. Each environment is distributed uniformly on $\mathcal{M}_{-i}$, and since $\mathcal{M}_{-i}$ has cardinality $\frac{\mu}{m_i}$,
\begin{equation}\label{eq:prodk}
\Pr\left[\mathbf{X}_{-i}^{h_i(k) - 1} \text{ for } k \in \{1,...,k^*_i(t)\} \text{ are all distinct} \right] = \prod_{k=1}^{k^*_i(t) - 1} \left( 1 - \frac{m_i k }{ \mu } \right) .
\end{equation}
If $k^*_i(t) > 1 + \frac{\mu}{m_i}$ then the probability in \eqref{eq:prodk} must be zero, and the lemma holds trivially ($k^*_i(t) > 1 + \frac{\mu}{m_i}$ implies $\floor*{\frac{t-i}{n}} > \frac{\mu}{m_i}$ which, in turn, implies $\ceil*{\frac{t}{n}}>\frac{\mu}{m_i}$, so the lower bound in the statement of the lemma is negative and the upper bound is positive). We will therefore consider the case in which $k^*_i(t) \leq 1 + \frac{\mu}{m_i}$.

We obtain the following upper bound:
\begin{align*}
\prod_{k=1}^{k^*_i(t)-1} \left( 1 - \frac{m_i k}{\mu} \right) &\leq \prod_{k=1}^{k^*_i(t)-1} \exp\left\{- \frac{m_i k}{ \mu }\right\} \leq \exp\left\{ - \frac{ m_i (k^*_i(t)-1)^2 }{2 \mu } \right\}  \leq \exp\left\{ - \frac{ m_i(\floor*{\frac{t}{n} - 1})^2 }{ 2\mu } \right\} .
\end{align*}
The first step follows from $ \exp\{x\}\geq1+x$ for all $x$. The final inequality follows from $k^*_i(t) - 1 = \floor*{\frac{t-i}{n}} \geq \floor*{\frac{t-n}{n}} = \floor*{\frac{t}{n} - 1}$.

We now turn to the lower bound:
\begin{align*}
\prod_{k=1}^{k^*_i(t) - 1} \left( 1 - \frac{m_i k}{\mu} \right) \geq 1 - \sum_{k=1}^{k^*_i(t) - 1} \frac{m_i k}{\mu} \geq1-  \frac{m_i}{\mu} \frac{k^*_i(t)^2}{2} \geq 1 - \frac{m_i}{\mu} \frac{(\ceil*{\frac{t}{n}})^2}{2} .
\end{align*}
The first step is an application of the Weierstrass product inequality. The final inequality follows from the fact that $k^*_i(t) = 1 + \floor*{\frac{t-i}{n}} \leq 1 + \floor*{\frac{t-1}{n}} = \ceil*{\frac{t}{n}}$.
\end{proof}

Define $m^* := \max_{i\in[n]} m_i$, so that $q_{n,\mathbf{m}} = \frac{\mu}{m^*}$.
\begin{proof}[Proof of Lemma \ref{prop:after_t}]
Recall that $T_{\langle \p{A},s_\texttt{c} \rangle}$ is the first time at which the path $\langle \p{A}, s_\texttt{c} \rangle$ hits a PNE or a best-response cycle. So $T_{\langle \p{A},s_\texttt{c} \rangle} > t$ is the event that $\langle \p{A},s_\texttt{c} \rangle$ hits PNE$(G_{n,\mathbf{m}})$ or a best-response cycle only after time $t$. It follows that
$$\Pr\left[\langle \p{A},s_\texttt{c} \rangle \text{ hits PNE$(G_{n,\mathbf{m}})$ after }t  \right] \leq \Pr\left[T_{\langle \p{A},s_\texttt{c} \rangle} > t \right] .$$
By Remarks \ref{rem:3}, \ref{rem:1}, and \ref{rem:2},
$$\Pr\left[T_{\langle \p{A},s_\texttt{c} \rangle} > t \right] \leq \Pr\left[F_{\langle \p{A},s_\texttt{c} \rangle} > t \right] = \Pr\left[F_{\langle \p{Y},s_\texttt{c} \rangle} > t \right]  = \Pr\left[F_{\langle \p{X},s_\texttt{c} \rangle} > t \right].$$
Now, let us focus on the path $\langle \p{X},s_\texttt{c} \rangle$ and consider a player $i$ satisfying $m_i = m^*$. The environments that player $i$ faces between times 1 and $t$ are given in the sequence $\mathbf{X}_{-i}^{h_i(1) - 1},\mathbf{X}_{-i}^{h_i(2) - 1},...,\mathbf{X}_{-i}^{h_i(k^*_i(t))-1}$. The event $F_{\langle \p{X},s_\texttt{c} \rangle} > t$ implies that the environments in this sequence are all distinct. Hence
$$ \Pr\left[F_{\langle \p{X},s_\texttt{c} \rangle} > t \right] \leq \Pr\left[\mathbf{X}_{-i}^{h_i(k) - 1} \text{ for } k \in \{1,...,k^*_i(t)\} \text{ are all distinct} \right] \leq \exp\left\{ - \frac{ (\floor*{\frac{t}{n} - 1})^2 }{ 2 q_{n,\mathbf{m}} } \right\},$$
where the final step follows from Lemma \ref{lem:bounds}.
\end{proof}

The proof of Lemma \ref{prop:after_t} illustrates why we study a coupled system. Finding an upper bound on the probability that $\langle \p{A},s_\texttt{c} \rangle$ hits PNE$(G_{n,\mathbf{m}})$ after $t$ is central to our proof of Theorem \ref{thm}. Our key step consists in showing that this probability is bounded above by the probability that the environments $\mathbf{X}_{-i}^{h_i(1) - 1}$, $\mathbf{X}_{-i}^{h_i(2) - 1}$,..., $\mathbf{X}_{-i}^{h_i(k^*_i(t))-1}$, which are generated by the clockwork random walk, are all distinct. This latter probability is easy to work out because the environments are independent uniform random draws. To avoid coupling, one might be tempted to argue that since the probability that $\langle \p{A},s_\texttt{c} \rangle$ hits PNE$(G_{n,\mathbf{m}})$ after $t$ is bounded above by the probability that the environments $\mathbf{A}_{-i}^{h_i(1) - 1},\mathbf{A}_{-i}^{h_i(2) - 1},...,\mathbf{A}_{-i}^{h_i(k^*_i(t))-1}$ generated by Algorithm \ref{alg2} are all distinct, one only needs to work out this latter probability. But this probability is not straightforward to work out: \emph{these} environments are not independent uniform random draws since they are generated by a path-dependent process.

To prove Lemma \ref{prop:by_t}, we introduce a slight modification of Algorithm \ref{alg5}. Algorithm \ref{alg6}, which describes a dynamic that is also coupled with the clockwork random walk, is identical to Algorithm \ref{alg5} except that for some particular profile $\mathbf{x}$ the algorithm is initialized with $R_i(\mathbf{x}_{-i})=x_i$ for all $i \in [n]$. This effectively ``plants'' a sink in the digraph (at $\mathbf{x}$).
\begin{algorithm}
\caption{Coupled dynamic with sink $\mathbf{x}$\label{alg6}} 
\texttt{
\begin{enumerate}[topsep=0pt,leftmargin=*]
\item Set $R_i(\mathbf{a}_{-i})=0$ for all $i \in [n]$ and $\mathbf{a}_{-i} \in \mathcal{M}_{-i}$
\item Set $R_i(\mathbf{x}_{-i})=x_i$ for all $i \in [n]$
\item Set the initial action profile to $\mathbf{Z}^0 = \mathbf{X}^0$
\item For $t \in \mathbb{N}$:
\begin{enumerate}
\item Set $i=s_\texttt{c}(t)$
\item Set $\mathbf{Z}_{-i}^{t} = \mathbf{Z}_{-i}^{t-1}$
\item If $R_i(\mathbf{Z}_{-i}^{t-1})=0$: set $Z_i^{t} = X_i^{t}$ and $R_i(\mathbf{Z}_{-i}^{t-1}) = Z_i^t$\\
If $R_i(\mathbf{Z}_{-i}^{t-1})\neq 0$: set $Z_i^{t} = R_i(\mathbf{Z}_{-i}^{t-1})$
\end{enumerate}
\end{enumerate}
}
\end{algorithm}
In the remaining steps, Algorithm \ref{alg6} selects a random initial profile and starts tracing a path by traveling along edges that (other than those edges already pointing to $\mathbf{x}$ in the initialization) are generated in an online manner. The paths traced by the clockwork random walk and this coupled dynamic with a sink at $\mathbf{x}$ must agree at least up to (but not including) the time at which \emph{either} an environment is encountered by a player for the second time \emph{or} the environment is $\mathbf{x}_{-i}$ for some player $i$.

$\langle \p{Z},\mathbf{x},s_\texttt{c} \rangle$ denotes a path generated according to Algorithm \ref{alg6}.

\begin{remark}\label{rem:4}
Consider an arbitrary action profile $\mathbf{x} \in \mathcal{M}$ and let $\langle \p{A}, s_\texttt{c} \rangle$ and $\langle \p{Z}, \mathbf{x}, s_\texttt{c} \rangle$ be generated according to Algorithms \ref{alg2} and \ref{alg6} respectively. Then the distribution of $\langle \p{A}, s_\texttt{c} \rangle$ \emph{conditional} on $\mathbf{x}$ being a pure Nash equilibrium, i.e.\ conditional on $\mathbf{x} \in \text{PNE}(G_{n,\mathbf{m}})$, is the same as the distribution of $\langle \p{Z}, \mathbf{x}, s_\texttt{c} \rangle$.
\end{remark}

\begin{proof}[Proof of Lemma \ref{prop:by_t}]
For any $t\in \mathbb{N}$,
\hypersetup{linkcolor=black}
\begingroup
\allowdisplaybreaks
\begin{align}
\Pr&\left[ \langle \p{A},s_\texttt{c} \rangle \text{ hits PNE$(G_{n,\mathbf{m}})$ by }t  \right] \nonumber \\
&=\sum_{\mathbf{x} \in \mathcal{M}} \Pr\left[ \langle \p{A},s_\texttt{c} \rangle \text{ hits } \{\mathbf{x}\} \text{ by }t \text{ and } \mathbf{x} \in \text{PNE}(G_{n,\mathbf{m}})  \right] \nonumber \\
&=\sum_{\mathbf{x} \in \mathcal{M}} \Pr\left[ \langle \p{A},s_\texttt{c} \rangle \text{ hits } \{\mathbf{x}\} \text{ by }t \, \Big| \, \mathbf{x} \in \text{PNE}(G_{n,\mathbf{m}}) \right] \Pr\left[\mathbf{x} \in \text{PNE}(G_{n,\mathbf{m}}) \right]  \nonumber \\
&=\sum_{\mathbf{x} \in \mathcal{M}} \underbrace{\Pr\left[ \langle \p{Z},\mathbf{x},s_\texttt{c} \rangle \text{ hits } \{\mathbf{x}\} \text{ by }t  \right]}_{(\ref{eq:last}.1)} \underbrace{\Pr\left[\mathbf{x} \in \text{PNE}(G_{n,\mathbf{m}}) \right]}_{(\ref{eq:last}.2)}  \label{eq:last}
\end{align}
\endgroup
\hypersetup{linkcolor=magenta}
The final step follows from Remark \ref{rem:4}; namely, the probability that $\langle \p{A},s_\texttt{c} \rangle$ hits $\{ \mathbf{x} \}$ by time $t$ conditional on $\mathbf{x} \in \text{PNE}(G_{n,\mathbf{m}})$ is equal to the probability that $\langle \p{Z},\mathbf{x}, s_\texttt{c} \rangle$ hits $\{\mathbf{x}\}$ by time $t$. We now analyze the expressions (\ref{eq:last}{\color{magenta}.1}) and (\ref{eq:last}{\color{magenta}.2}).

For (\ref{eq:last}{\color{magenta}.2}), since payoffs are drawn identically and independently according to the atomless distribution $\mathbb{P}$, we have that
\begin{equation}
\Pr\left[\mathbf{x} \in \text{PNE}(G_{n,\mathbf{m}})  \right] = \prod_{i=1}^n \Pr\left[U_{i}\left(\mathbf{x} \right) \geq \max_{x_i' \in [m_i]} U_{i} \left(x_i',\mathbf{x}_{-i}\right) \right] = \frac{1}{\mu} . \label{eq:prod1}
\end{equation}

We now find upper and lower bounds on (\ref{eq:last}{\color{magenta}.1}) by relating $\langle \p{Z}, \mathbf{x}, s_\texttt{c} \rangle$ to the clockwork random walk path $\langle \p{X}, s_\texttt{c} \rangle$. We start with the upper bound. Notice that $\langle \p{Z}, \mathbf{x}, s_\texttt{c} \rangle$ cannot hit $\{\mathbf{x}\}$ by time $t$ unless $\mathbf{X}^{\tau-1}_{-s_\texttt{c}(\tau)} = \mathbf{x}_{-s_\texttt{c}(\tau)} $ for some $\tau \leq t$. Therefore
\begin{align}
\Pr\left[\langle \p{Z}, \mathbf{x}, s_\texttt{c} \rangle \text{ hits }\{\mathbf{x}\} \text{ by }t \right] &\leq  \Pr\left[\bigcup_{\tau=1}^t \{ \mathbf{X}^{\tau-1}_{-s_\texttt{c}(\tau) \}} = \mathbf{x}_{-s_\texttt{c}(\tau)} \} \right] \nonumber \\
&\leq \sum_{\tau=1}^t \Pr\left[ \mathbf{X}^{\tau-1}_{-s_\texttt{c}(\tau)} = \mathbf{x}_{-s_\texttt{c}(\tau)} \right] \nonumber\\
&= \sum_{\tau=1}^t\frac{ m_{s_\texttt{c}(\tau)}}{\mu} \nonumber \\
&\leq \frac{t}{q_{n,\mathbf{m}}} . \label{eq:up1}
\end{align}
The penultimate step follows from the fact that $\mathbf{X}^{\tau-1}_{-s_\texttt{c}(\tau)}$ consists of $n-1$ independent uniform random variables (one action for each player other than $s_\texttt{c}(\tau)$), so $\mathbf{X}^{\tau-1}_{-s_\texttt{c}(\tau)}$ is itself uniformly drawn from $\mathcal{M}_{-s_\texttt{c}(\tau)}$, and $\mathcal{M}_{-s_\texttt{c}(\tau)}$ has cardinality $\frac{\mu}{m_{s_\texttt{c}(\tau)}}$.

We now turn to the lower bound. If $F_{\langle \p{X},s_\texttt{c} \rangle} > t$ and $\mathbf{X}^{\tau-1}_{-s_\texttt{c}(\tau)} = \mathbf{x}_{-s_\texttt{c}(\tau)} $ for some $\tau \leq t$ then $\langle \p{Z}, \mathbf{x}, s_\texttt{c} \rangle$ must hit $\{\mathbf{x}\}$ by time $t$. In other words, if no environments are repeated for any player and the environment is $\mathbf{x}_{-i}$ for some player $i$ by time $t$, then $\langle \p{Z}, \mathbf{x}, s_\texttt{c} \rangle$ must hit $\{\mathbf{x}\}$ by time $t$. Therefore,
\begingroup
\allowdisplaybreaks
\begin{align}
 \Pr\left[\langle \p{Z}, \mathbf{x}, s_\texttt{c} \rangle \text{ hits }\{\mathbf{x}\} \text{ by } t \right] \geq &  \Pr\left[\bigcup_{\tau=1}^t \{\mathbf{X}^{\tau-1}_{-s_\texttt{c}(\tau)} = \mathbf{x}_{-s_\texttt{c}(\tau)}\} \text{ and } F_{\langle \p{X},s_\texttt{c} \rangle} > t \right]  \nonumber \\
 =& \Pr\left[\bigcup_{\tau=1}^t \{\mathbf{X}^{\tau-1}_{-s_\texttt{c}(\tau)} = \mathbf{x}_{-s_\texttt{c}(\tau)}\} \, \bigg| \, F_{\langle \p{X},s_\texttt{c} \rangle} > t \right] \Pr\left[F_{\langle \p{X},s_\texttt{c} \rangle} > t \right] . \label{eq:condition}
 \end{align}
 \endgroup

 To bound the first term in \eqref{eq:condition}, select a player $i$ satisfying $m_i = m^*$ and notice that $\mathbf{X}_{-i}^{h_i(k) - 1} = \mathbf{x}_{-i}$ for some $k \in \{1,...,k^*_i(t)\}$ implies that $\mathbf{X}^{\tau-1}_{-s_\texttt{c}(\tau)} = \mathbf{x}_{-s_\texttt{c}(\tau)} $ for some $\tau \leq t$. Therefore
 \begin{align}
\Pr\left[\bigcup_{\tau=1}^t \{\mathbf{X}^{\tau-1}_{-s_\texttt{c}(\tau)}  = \mathbf{x}_{-s_\texttt{c}(\tau)}\} \, \bigg| \, F_{\langle \p{X},s_\texttt{c} \rangle} > t \right] &\geq \Pr\left[ \bigcup_{k=1}^{k^*_i(t)} \{ \mathbf{X}_{-i}^{h_i(k) - 1} = \mathbf{x}_{-i} \}  \, \bigg| \, F_{\langle \p{X},s_\texttt{c} \rangle} > t\right]  \nonumber \\
&= \sum_{k=1}^{k^*_i(t)} \Pr\left[  \mathbf{X}_{-i}^{h_i(k) - 1} = \mathbf{x}_{-i}   \, \Big| \, F_{\langle \p{X},s_\texttt{c} \rangle} > t \right] \nonumber \\
&= \sum_{k=1}^{k^*_i(t)} \frac{m^*}{\mu} \nonumber \\
&\geq \floor*{\frac{t}{n}} \frac{1}{q_{n,\mathbf{m}}} . \label{eq:down1}
\end{align}
The second line follows from the fact that since all the environments for our chosen player $i$ are distinct, the events in the union are mutually exclusive. The next step follows from the fact that our process is invariant under symmetry. So for any $k \in \{1,...,k^*_i(t)\}$ and for all $\mathbf{x}_{-i}$ and $\mathbf{y}_{-i}$, $\Pr[ \mathbf{X}_{-i}^{h_i(k) - 1} = \mathbf{x}_{-i}  \, | \, F_{\langle \p{X},s_\texttt{c} \rangle} > t] = \Pr[ \mathbf{X}_{-i}^{h_i(k) - 1} = \mathbf{y}_{-i}  \, | \, F_{\langle \p{X},s_\texttt{c} \rangle} > t]$ which implies that $\Pr[ \mathbf{X}_{-i}^{h_i(k) - 1} = \mathbf{x}_{-i}  \, | \, F_{\langle \p{X},s_\texttt{c} \rangle} > t] = \frac{m^*}{\mu} = \frac{1}{q_{n,\mathbf{m}}}$. The last step follows from $k^*_i(t) = 1 + \floor*{\frac{t-i}{n}} \geq 1 + \floor*{\frac{t}{n} - 1} = \floor*{\frac{t}{n}}$.

To bound the second term in \eqref{eq:condition}, notice that if for each $i\in [n]$ the environments $\mathbf{X}_{-i}^{h_i(1) - 1}, \mathbf{X}_{-i}^{h_i(2) - 1},...,\mathbf{X}_{-i}^{h_i(k^*_i(t)) - 1}$ are all distinct then $F_{\langle \p{X},s_\texttt{c} \rangle} > t$. Therefore
\begingroup
\allowdisplaybreaks
  \begin{align}
 \Pr[F_{\langle \p{X},s_\texttt{c} \rangle} > t] &\geq \Pr\left[\bigcap_{i \in [n]} \{\mathbf{X}_{-i}^{h_i(k) - 1} \text{ for } k \in \{1,...,k^*_i(t)\} \text{ are all distinct} \} \right] \nonumber \\
 &= 1 - \Pr\left[\bigcup_{i \in [n]} \{\mathbf{X}_{-i}^{h_i(k) - 1} \text{ for } k \in \{1,...,k^*_i(t)\} \text{ are not all distinct} \} \right]  \nonumber\\
 & \geq 1 - \sum_{i \in [n]} \Pr\left[\mathbf{X}_{-i}^{h_i(k) - 1} \text{ for } k \in \{1,...,k^*_i(t)\} \text{ are not all distinct}\right] \nonumber\\
 & \geq 1- \frac{(\ceil*{\frac{t}{n}})^2}{2} \frac{\sum_{i=1}^n m_i}{\mu} \nonumber \\
& \geq 1- \frac{(\ceil*{\frac{t}{n}})^2}{2} \frac{n}{q_{n,\mathbf{m}}} . \label{eq:down2}
 \end{align}
 \endgroup
The penultimate step follows from Lemma \ref{lem:bounds}.

Gathering the results \eqref{eq:last}, \eqref{eq:prod1}, \eqref{eq:up1}, \eqref{eq:down1}, and \eqref{eq:down2} together yields the desired conclusion.
\end{proof}

\subsection{Results for 2-player games}
In games with $n=2$ players, the action taken by player $s_{\texttt{c}}(t)$ at $t$ corresponds exactly to the environment that player $s_{\texttt{c}}(t+1)$ faces at $t+1$. We take advantage of this property in our proof of Theorem \ref{thm2} below.

\begin{proof}[Proof of Theorem \ref{thm2}]
Let $\eta_t$ denote the probability under the $2$-player clockwork random walk that, by time $t$, no player plays an action that corresponds to an environment that was ever encountered by the other player. For $t \geq 1$ we have
$$ 
\eta_{t+1} = \eta_{t} \times \left( 1 - \frac{\ceil*{\frac{t}{2}}}{m_{s_\texttt{c}(t+1)}}\right).
$$
The term in parentheses is the probability that player $s_\texttt{c}(t+1)$ does not repeat any of the $\ceil*{\frac{t}{2}}$ environments encountered by player $s_\texttt{c}(t)$ by time $t$. Solving with $\eta_1=1$ yields
$$
\eta_t = \prod_{i=1}^{t} \left(1 - \frac{1}{m_{s_\texttt{c}(i)}} \floor*{\frac{i}{2}} \right),
$$
and, evidently, $\eta_t$ is non-negative provided $t \leq 2 m_*$.

For the path to hit a 2$k$-cycle at time $t$, it must be that (i) by time $t+2k-2$, no player plays an action that corresponds to an  environment that was ever encountered by the other player, and (ii) the action taken by player $s_\texttt{c}(t+2k-1)$ at time $t+2k-1$ is equal to the environment encountered by player $s_\texttt{c}(t)$ at time $t$. So, the probability that the clockwork sequence best-response dynamic converges to a $2k$-cycle at time $t$ is $\frac{\eta_{t+2k-2}}{m_{s_\texttt{c}(t+2k-1)}}$, which completes the proof.
\end{proof}

For the remaining proofs, we employ the following standard notation for asymptotics: we write $f(n) = o(g(n))$ if $f(n)/g(n) \rightarrow 0$ as $n \rightarrow \infty$, $f(n) \sim g(n)$ if $f(n)/g(n) \rightarrow 1$ as $n \rightarrow \infty$, and $f(n) =O( g(n)) $ if there is $M>0$ and $N$ such that $| f(n)| \leq Mg(n)$ for all $n \geq N$.

\begin{proof}[Proof of Theorem \ref{thm3}]
$\eta_t$ is precisely the probability that the clockwork best-response dynamic does not hit a 2$k$-cycle (for any $k$) until at least time $t$. With $m_1=m_2=m$ we can write $\eta_t$ as
\begin{equation*}
\prod_{i=1}^{t} \left(1 - \frac{1}{m} \floor*{\frac{i}{2}} \right) = 
\begin{cases}
\frac{m!^2}{( m - \frac{t+1}{2} )!^2 m^{t+1} }  & \text{ if $t$ is odd} \\
\left( \frac{m - \frac{t}{2}}{m} \right) \frac{m!^2}{( m - \frac{t}{2} )!^2 m^{t} } & \text{ if $t$ is even}
\end{cases} .
\end{equation*}
Using Stirling's formula which states that
$$n! \sim \sqrt{2\pi n} \cdot n^n \exp\{-n\} , $$
as $n \rightarrow \infty$, we obtain
\begin{equation}\label{eq:stirling1}
\frac{m!^2}{( m - \frac{t+1}{2} )!^2 m^{t+1} } \sim \left( \frac{m - \frac{t+1}{2}}{m} \right)^{t -2m} \exp \{ - (t+1) \} ,
\end{equation}
and
\begin{equation}\label{eq:stirling2}
\left( \frac{m - \frac{t}{2}}{m} \right) \frac{m!^2}{( m - \frac{t}{2} )!^2 m^{t} } \sim \left( \frac{m - \frac{t}{2}}{m} \right)^{t -2m} \exp \{ - t \} ,
\end{equation}
whenever $m - t \rightarrow \infty$. Taking a logarithm of the last expression we get
\begin{align*}
- t + (t -2m)\ln\left( 1 - \frac{1}{m} \frac{t}{2} \right) &=- t + (t -2m)\left( - \frac{1}{2} \frac{t}{m}  - \frac{1}{8} \frac{t^2}{m^2}  + O\left( \frac{t^3}{m^3}\right) \right) \\
&=  -\frac{1}{4} \frac{t^2}{m} + O \left( \frac{t^3}{m^2} \right).
\end{align*}
Provided that $t = o(m^{2/3})$, the second term goes to zero and therefore equation \eqref{eq:stirling2} behaves asymptotically like $ \exp\{ - t^2 / (4m)\}$. An identical argument shows that, under the same conditions, \eqref{eq:stirling1} is also asymptotically $ \exp\{ - t^2 / (4m)\}$. Hence,
\begin{equation}\label{eq:approx1}
\prod_{i=1}^{t} \left(1 - \frac{1}{m} \floor*{\frac{i}{2}} \right) \sim \exp \left\{ - \frac{t^2}{4m} \right\} .
\end{equation}
This completes the proof of Theorem \ref{thm3}. Note that approximation \eqref{eq:approx1} holds uniformly in the range $[1  , o(m^{2/3})]$.
\end{proof}

\begin{proof}[Proof of Proposition \ref{prop:F1}]
Let $T=T(m)$ satisfy $T = o(m^{2/3})$ and $k = o(T)$. We assume that $T \geq \frac{m^{2/3}}{\ln(m)}$ so that $T$ is not too small, and we split the summation in \eqref{eq:n=2} into two ranges: $t \leq T$ and $t > T$. Since \eqref{eq:approx1} holds uniformly in our first range, we have
\begin{align*}
\frac{1}{m} \sum_{t=1}^{T}  \prod_{i=1}^{t+2(k-1)} \left( 1 - \frac{1}{m} \floor*{\frac{i}{2}}  \right) \sim  \frac{1}{m} \sum_{t=1}^{T} \exp \left\{ -  \frac{( t + 2(k-1) )^2}{4m} \right\}.
\end{align*}
We now approximate the summation on the right-hand side with an integral. Firstly, note that
\begin{align}
\frac{1}{m} \int_{1}^{T+1} \exp \left\{ -  \frac{( t + 2(k-1) )^2}{4m} \right\} dt &= \sqrt{\frac{2}{m}} \int_{\frac{2k -1 }{\sqrt{2m}}}^{\frac{T+1 + 2(k-1)}{\sqrt{2m}}} \exp \left\{ -  \frac{x^2}{2} \right\} dx \nonumber \\
&\sim  \sqrt{\frac{2}{m}} \int_{\frac{2k -1 }{\sqrt{2m}}}^{\infty} \exp \left\{ -  \frac{x^2}{2} \right\} dx \nonumber \\
&=2 \sqrt{\frac{\pi}{m}} \left( 1 - \Phi\left( \frac{2k -1}{\sqrt{2m}}  \right)\right) , \label{eq:lim_int}
\end{align}
where the first step uses the transformation $x = (t + 2(k-1))/\sqrt{2m}$. Furthermore,
\begin{align*}
\frac{1}{m} \int_{0}^1 \exp \left\{ -  \frac{( t + 2(k-1) )^2}{4m} \right\} dt \leq \frac{1}{m} ,
\end{align*}
which goes to zero faster than \eqref{eq:lim_int}. Since
$$\int_1^{T+1} f(t) dt \leq \sum_{t=1}^T f(t) \leq \int_{0}^T f(t) dt \leq \int_1^{T+1} f(t) dt + \int_{0}^1 f(t) dt ,$$
for any positive and decreasing function $f(\cdot)$, it follows that
$$\frac{1}{m} \sum_{t=1}^{T} \exp \left\{ -  \frac{( t + 2(k-1) )^2}{4m} \right\} \sim  2 \sqrt{\frac{\pi}{m}} \left( 1 - \Phi\left( \frac{2k -1}{\sqrt{2m}}  \right)\right). $$
It remains for us to show that the summation \eqref{eq:n=2} over the second range is negligible. Since $\exp\{x\} \geq 1+x$ and $\floor{x} > x -1$ for all $x$, we obtain the following upper bound:
\begingroup
\allowdisplaybreaks
\begin{align*}
\frac{1}{m} \sum_{t=T+1}^{2(m-k+1)} \;  \prod_{i=1}^{t+2(k-1)} \left( 1 - \frac{1}{m} \floor*{\frac{i}{2}}  \right)  &\leq \frac{1}{m} \sum_{t=T+1}^{2(m-k+1)} \; \prod_{i=1}^{T+1+2(k-1)} \left( 1 - \frac{1}{m} \floor*{\frac{i}{2}}  \right) \\
&\leq \frac{1}{m} \sum_{t=T+1}^{2(m-k+1)} \exp \left\{ - \frac{1}{m} \sum_{i=1}^{T+1 + 2(k-1)}\left( \frac{i}{2} -1 \right) \right\} \\
&\leq \frac{2m -2k - T+1}{m} \exp \left\{ - \frac{1}{4m} \left(T + 2(k-1) -2  \right)^2  \right\} .
\end{align*}
\endgroup
This expression is small compared to the other half of the sum.
\end{proof}

\def\bibfont{\footnotesize}
\bibliographystyle{chicago}
\bibliography{references}

\end{document}